\documentclass[11pt]{llncs} 

\usepackage{array}
\usepackage{amsmath}
\usepackage{mathtools}
\usepackage{graphics}
\usepackage{epsfig}
\usepackage{latexsym}
\usepackage{amssymb}
\usepackage{color}
\usepackage[ruled,lined,algonl]{algorithm2e}
\usepackage{enumerate}
\usepackage{dsfont}

\usepackage{longtable}
\usepackage[margin=2.54cm]{geometry}

\newcommand{\Hom}{\mathbf{H}}
\newcommand{\Cycles}{\mathbf{Z}}
\newcommand{\Boundaries}{\mathbf{B}}
\newcommand{\Chains}{\mathbf{C}}

\newcommand{\C}{\mathbb{C}}
\newcommand{\co}{\colon\thinspace}

\newcommand{\tri}{\mathfrak{T}}
\newcommand{\M}{M}

\newcommand{\tv}{\mathrm{TV}}
\newcommand{\adm}{\mathrm{Adm}}

\newcommand{\Z}{\mathbb{Z}}
\newcommand{\zmin}{z^-}
\newcommand{\zmax}{z^+}
\newcommand{\Q}{\mathbb{Q}}

\newcommand{\X}{\mathbb{X}}
\newcommand{\Y}{\mathbb{Y}}
\newcommand{\disk}{\mathbb{D}}

\newcommand{\intsymb}[3]{ [#1 \ \ #2 \ \ #3] }

\newcommand{\mcg}{\mathcal{MCG}}
\newcommand{\id}{\mathrm{id}}

\newcommand{\SL}{sl}

\begin{document}

\title{Admissible colourings of 3-manifold triangulations for Turaev-Viro type 
invariants}

\author{Cl\'ement Maria \and Jonathan Spreer
  }

\institute{The University of Queensland, Australia\\
   \email{c.maria@uq.edu.au, j.spreer@uq.edu.au}
}

\maketitle

\begin{abstract}
Turaev Viro invariants are amongst the most powerful tools to distinguish
$3$-manifolds: They are implemented in mathematical software, and allow 
practical computations. The invariants can be computed purely combinatorially by 
enumerating colourings on the edges of a triangulation $\tri$. 

These edge colourings can be interpreted as embeddings of surfaces in $\tri$. 
We give a characterisation of how these embedded surfaces intersect with 
the tetrahedra of $\tri$. This is done by characterising isotopy classes of 
simple closed loops in the $3$-punctured disk. As a direct result we obtain a 
new system of coordinates for edge colourings which allows for simpler 
definitions of the tetrahedron weights incorporated in the Turaev-Viro 
invariants. 

Moreover, building on a detailed analysis of the colourings, as well as
classical work due to Kirby and Melvin, Matveev, and others, we show that 
considering a much smaller set of colourings suffices to compute
Turaev-Viro invariants in certain significant cases. This results in a 
substantial improvement of running times to compute the invariants, 
reducing the number of colourings to consider by a factor of $2^n$. In addition,
we present an algorithm to compute Turaev-Viro invariants of degree four -- a 
problem known to be $\#P$-hard -- which capitalises on the combinatorial structure
of the input.

The improved algorithms are shown to be optimal in the following sense: There 
exist triangulations admitting all colourings the algorithms consider. 
Furthermore, we demonstrate that our new algorithms to compute Turaev-Viro 
invariants are able to distinguish the majority of $\Z$-homology spheres with 
complexity up to $11$ in $O(2^n)$ operations in $\Q$.
\end{abstract}


\section{Introduction}
\label{sec:intro}

In geometric topology, recognising the topological type of a given manifold,
i.e., testing if two manifolds are equivalent, is one of the most fundamental 
algorithmic problems. In fact, the task of comparing the topology of two given 
manifolds often stands at the very beginning of a question, and solving 
it is essential for conducting research in the field.

Depending on the dimension of the manifolds, this task is remarkably
difficult to solve in general. There exists an algorithmic solution in 
dimension three due to Perelman's proof of the geometrisation 
conjecture \cite{kleiner08-perelman}, but it is highly theoretical in
nature and has never been implemented. Moreover, in dimensions $\geq 5$ the 
problem becomes undecidable \cite{Stillwell93Undecidability}.

As a result, comparing the topology of two manifolds in dimension $\geq 3$
often requires human-machine interactions, combining various
strategies. In practice, these are largely of two distinct types: (i) 
computing {\em invariants} in order to prove that two given manifolds are 
distinct; and (ii) trying to establish a certificate that two 
given manifolds are homeomorphic.

Here we focus on the former type of methods, more precisely, on a particularly 
powerful family of invariants for $3$-manifolds, called the {\em Turaev-Viro 
invariants} \cite{turaev92-invariants}. These are parameterised by two integers 
$r$ and $q$, with $r \geq 3$, $0 < q < 2r$, and denoted by $\tv_{r,q}$. They 
derive from quantum field theory but can be computed by purely combinatorial 
means---much like the famous Jones polynomial for knots. Implementations exists
for $3$-manifolds represented by (i) {\em spines} ($2$-dimensional skeletons of 
$3$-manifolds) in Matveev's \emph{Manifold Recogniser} 
\cite{matveev03-algms,recogniser}; and (ii) {\em triangulations} (a 
list of tetrahedra with their triangular faces glued together in pairs) in
Burton's software \emph{Regina} \cite{regina,Burton15TuraevViro}. 

We work within the latter setting, namely {\em 
triangulations of $3$-manifolds} $\tri$, where the definition of the Turaev-Viro 
invariant $\tv_{r,q}$ is based on \emph{admissible colourings} of the edges of 
$\tri$ with $r-1$ distinct colours. Each admissible colouring defines a {\em 
weight}, and $\tv_{r,q} (\tri)$ equals the sum of these weights. The naive 
implementation of this procedure is simple, efficient in memory, but has worst 
case running time $(r-1)^{n+v}$, where $n$ is the number of tetrahedra and $v$ 
the number of vertices of $\tri$. More recently, Burton and the authors 
introduced a {\em fixed parameter tractable (FPT)} algorithm which is {\em 
linear} in $n$, and only singly exponential in the treewidth of the dual graph of 
$\tri$ \cite{Burton15TuraevViro}.

%
In this article, we study admissible colourings of $3$-manifold triangulations
with the aim of a better understanding of Turaev-Viro invariants and significant
algorithmic improvements.

Admissible colourings can be interpreted as surfaces \emph{embedded}
in a triangulated $3$-manifold, where the colour of an edge corresponds to (half) 
its number of intersections with the surface, and the surface intersects the
triangles of the triangulation in straight arcs. Embedded surfaces play a
vital role in $3$-manifold topology, most notably due to Haken's theory of
\emph{normal surfaces}, i.e. embedded surfaces intersecting the 
tetrahedra of a triangulation in a collection of triangles and quadrilaterals 
\cite{haken61-knot}. Surfaces of critical importance to the topology of a 
manifold, such as a disk bounding an unknot in a triangulation of a 
knot-complement, can be proven to occur as a normal surface. For other problems,
such as recognising the $3$-sphere, surfaces of slightly more general types have
to be considered \cite{rubinstein97-3sphere}. Surfaces coming from 
admissible colourings contain all these surface types and many more (depending 
on the value of $r$). See recent work by Bachman \cite{Bachmann12Helical} for a 
study of such {\em surfaces of arbitrary index} from a topological point of 
view. This illustrates the potential power of the Turaev-Viro 
invariants in distinguishing between non-homeomorphic $3$-manifolds, based on 
purely combinatorial objects. 

We present a classification of surface types defined by admissible
colourings in form of isotopy classes of simple closed loops in the 
$3$-punctured disk (as a model of the tetrahedron). 
We give a combinatorial characterisation of this 
bijection using intersection numbers of the surface with the six edges 
of a tetrahedron. As an application of this characterisation we transform and 
simplify the formulae for the tetrahedra weights for $\tv_{r,q}$, 
in terms of the surface pieces intersecting a tetrahedron.

Moreover, we build on work by Kirby and 
Melvin~\cite{kirby91-witten} 
and Matveev~\cite{matveev03-algms} 
to bound the
number of admissible colourings relative to the size, the number of vertices, 
and the first Betti number of a triangulation. In particular, we prove sharp 
upper bounds on the number of admissible colourings which are much smaller than 
the trivial upper bound, and which are strongest 
for triangulations of homology spheres with only one vertex. 
In addition, we obtain a new upper bound for the size of $\adm (\tri, 4)$
depending on the specific structure of the input triangulation, which is sharp
in many cases. Note that this is particularly interesting since computing 
$\tv_{4,1}$ is a $\#P$-hard problem~\cite{Burton15TuraevViro,kirby04-nphard}.

We use these bounds together with classical results from $3$-manifold topology 
to obtain a significant exponential speed-up of the computation of some 
Turaev-Viro invariants. In particular, we reduce running times of the naive 
enumeration procedure from $(r-1)^{n+v}$ to $O(\lceil (r-1)/2 \rceil^{n+1})$
for $r$ odd and $q=1$; 
and describe an enumeration procedure to compute 
$\tv_{4,1}$ which is shown to be 
near-optimal in most cases of small $3$-manifold triangulations.


Note that the improved algorithms still have exponential running times. 
However, we give experimental evidence that the reduction in the base of the 
exponent is of practical significance for
triangulations of intermediate size
, and $4 \leq r \leq 9$. 
This range of values is highly relevant for major applications such as building 
censuses of minimal triangulations.

\section{Background}
\label{sec:bg}

\subsection{Manifolds, triangulations, and (co-)homology groups}
\label{ssec:trigs}

Let $M$ be a closed $3$-manifold. A \emph{generalised triangulation} $\tri$ of 
$M$ is a collection of $n$ abstract tetrahedra $\Delta_1,\ldots,\Delta_n$ 
together with $2n$ {\em gluing maps} identifying their $4n$ triangular faces 
in pairs, such that the underlying topological space is homeomorphic to $M$.

As a consequence of the gluings, vertices, edges or triangles of 
the same tetrahedron may be identified. Indeed, it is common in practical 
applications to have a \emph{one-vertex triangulation}, in which all vertices of
all tetrahedra are identified to a single point. We refer to an equivalence 
class defined by the gluing maps as a single \emph{face of the triangulation}.

Generalised triangulations are widely used in $3$-manifold topology. They are 
closely related, but more general, than simplicial complexes: every simplicial
complex triangulating a manifold is a generalised triangulation, and some 
subdivision of a generalised triangulation is always a simplicial
complex.

Let $\tri$ be a generalised $3$-manifold triangulation. For the field of 
coefficients 
$\Z_2 := \Z / 2\Z$, the {\em group of $p$-chains}, $0 \leq p \leq 3$, denoted 
$\Chains_p(\tri,\Z_2)$, of $\tri$ is the group of formal sums of $p$-faces with 
$\Z_2$ coefficients. The \emph{boundary operator} is a linear operator 
$\partial_p: \Chains_p(\tri,\Z_2) \rightarrow 
\Chains_{p-1}(\tri,\Z_2)$ such that $\partial_p \sigma = \partial_p \{v_0, 
\cdots , v_p\} = \sum_{j=0}^p \{v_0,\cdots ,\widehat{v_j}, \cdots,v_p\}$,
where $\sigma$ is a face of $\tri$, $\{v_0, \ldots, v_p\}$ represents 
$\sigma$ as a face of a tetrahedron of $\tri$ in local vertices $v_0, 
\ldots, v_p$, and $\widehat{v_j}$ means $v_j$ is deleted from the list. Denote 
by $\Cycles_p(\tri,\Z_2)$ and $\Boundaries_{p-1}(\tri,\Z_2)$ the kernel and the 
image of $\partial_p$ respectively. Observing $\partial_p \circ 
\partial_{p+1}=0$, we define the {\em $p$-th homology group} $\Hom_p(\tri,\Z_2)$ of 
$\tri$ by the quotient $\Hom_p(\tri,\Z_2) = \Cycles_p(\tri,\Z_2)/
\Boundaries_p(\tri,\Z_2)$. These structures are vector spaces. Informally, the $p$-th homology group, $0 \leq p \leq 3$, of a generalised 
triangulation $\tri$ counts the number of ``$p$-dimensional holes'' in $\tri$. 

The concept of {\em cohomology} is in many ways dual to homology, but more 
abstract and endowed with more algebraic structure. It is defined in the 
following way: The {\em group of $p$-cochains} $\Chains^p(\tri,\Z_2)$ is the
formal sum of linear maps of $p$-faces of $\tri$ into $\Z_2$. The 
\emph{coboundary operator} is a linear operator $\delta^p: 
\Chains^{p-1}(\tri,\Z_2) \rightarrow \Chains^{p}(\tri,\Z_2)$ such that 
for all $\phi \in \Chains^{p-1}(\tri,\Z_2)$ we have $\delta^p (\phi) = 
\phi \circ \partial_p$. As above, {\em $p$-cocycles} are the elements in the
kernel of $\delta^{p+1}$, {\em $p$-coboundaries} are elements in the image
of $\delta^{p}$, and the \emph{$p$-th cohomology group} 
$\Hom^p(\tri,\Z_2)$ is defined as the $p$-cocycles factored by the 
$d$-coboundaries.

The exact correspondence between elements of homology and cohomology is best
illustrated by {\em Poincar\'e duality} stating that for closed $d$-manifold 
triangulations $\tri$, $\Hom^p(\tri,\Z_2)$ and $\Hom_{d-p}(\tri,\Z_2)$ are
dual as vector spaces. More precisely, let $S$ be a $2$-cycle in $\tri$ 
representing a class in $\Hom_{2}(\tri,\Z_2)$.
We can perturb $S$ such that it contains no vertex of $\tri$ and intersects
every tetrahedron of $\tri$ in a single triangle (separating one vertex from the
other three) or a single quadrilateral (separating pairs of vertices).
It follows that every edge of $\tri$ intersects $S$ in $0$ or $1$ points.
Then the $1$-cochain defined by mapping every edge intersecting $S$ to $1$ and
mapping all other edges to $0$ represents the Poincar\'e dual of $S$ in 
$ \Hom^1(\tri,\Z_2) $.

In this article we will mostly consider the first cohomology group $\Hom^1(\tri,\Z_2)$ 
---a $\Z_2$-vector space of dimension called the {\em first Betti number
 $\beta_1(\tri,\Z_2)$ of $\tri$}.
(Co)homology groups can be computed on a triangulation in polynomial time. 
For a more detailed introduction into (co)homology theory see 
\cite{hatcher02-algebraic}.

\subsection{Turaev-Viro type invariants}
\label{ssec:tv}

Let $\tri$ be a generalised triangulation of a closed $3$-manifold $\M$,
and let $r \geq 3$, be an integer. Let $V$, $E$, $F$ and $T$ denote the set
of vertices, edges, triangles and tetrahedra of the triangulation $\tri$
respectively. Let $I = \{0, 1/2, 1, 3/2, \ldots, (r-2)/2\}$ be the set of the
first $r-1$ non-negative half-integers. 
A \emph{colouring} of $\tri$ is defined to be a map $\theta\co E \to I$;
that is, $\theta$ ``colours'' each edge of $\tri$ with an element
of $I$. A colouring $\theta$ is \emph{admissible} if, for each triangle of 
$\tri$, the three edges $e_1$, $e_2$, and $e_3$ bounding the triangle satisfy 
the 
\begin{itemize}
  \item \emph{parity condition} $\theta(e_1)+\theta(e_2)+\theta(e_3)\in \Z$;
  \item \emph{triangle inequalities} $\theta(e_i) \leq \theta(e_j) + 
    \theta(e_k)$, $\{i,j,k\} = \{ 1,2,3\}$; and 
  \item \emph{upper bound constraint} $\theta(e_1)+\theta(e_2)+\theta(e_3)\leq 
    r-2$.
\end{itemize}
For a triangulation $\tri$ and a value $r \geq 3$, its set of 
admissible colourings is denoted by $\adm (\tri,r)$.

For each admissible colouring $\theta$ and for each vertex $u \in V$, edge 
$e \in E$, triangle $f \in F$ or tetrahedron $t \in T$ we define 
\emph{weights} $|u|_{\theta}, |e|_{\theta}, |f|_{\theta}, 
|t|_{\theta} \in \C$. The 
weights of vertices are constant, and the weights of edges, triangles and 
tetrahedra only depend on the colours of edges they are incident to. 
Using these weights, we define the \emph{weight of the colouring} to be
\begin{equation}
\label{eq:colouring}
|\tri|_{\theta} =
    \prod_{v \in V} |u|_{\theta} \times
    \prod_{e \in E} |e|_{\theta} \times
    \prod_{f \in F} |f|_{\theta} \times
    \prod_{t \in T} |t|_{\theta},
\end{equation}

\emph{Invariants of Turaev-Viro types} of $\tri$ are defined as sums of the 
weights of all admissible colourings of $\tri$, that is
$\tv_{r}(\tri) = \sum_{\theta \in \adm(\tri,r)} |\tri|_{\theta}$.

In \cite{turaev92-invariants}, Turaev and Viro show that, when the weighting 
system satisfies some identities, $\tv_{r}(\tri)$ is indeed an invariant of 
the manifold; that is, if $\tri$ and $\tri'$ are generalised triangulations of 
the same closed 3-manifold $\M$, then $\tv_{r}(\tri) = \tv_{r}(\tri')$ for 
all $r$. We will thus sometimes abuse notation and write 
$\tv_{r} (\M)$, meaning the Turaev-Viro invariant computed for a triangulation
of $\M$.
In Section~\ref{app:tv-weights} we give the precise definition of the weights 
of the original Turaev-Viro invariant at $\SL_2(\C)$, which not only
depend on $r$ but also on a second integer $0 < q < 2r$. The definition of these 
weights is rather involved, and the study of admissible colourings in 
Section~\ref{sec:col-tetrahedron} allows us to give more comprehensible 
formulae.

For an $n$-tetrahedra triangulation $\tri$ with $v$ vertices (and thus, 
necessarily $n+v$ edges), there is a simple backtracking algorithm to compute 
$\tv_r(\tri)$ by testing the $(r-1)^{v+n}$ possible colourings for admissibility 
and computing their weights. The case $r=3$ can however be computed in 
polynomial time, due to a connection between $\tv_3$ and 
homology~\cite{Burton15TuraevViro,matveev03-algms}. 

\subsection{Turaev-Viro invariants at a cohomology class}
\label{ssec:tvhom}

Let $\Hom^1(\tri, \Z_2) = (\Z_2)^{\beta_1 (\tri, \Z_2)}$ be the cohomology group
of $\tri$ in dimension one with $\Z_2$ coefficients, and let $\alpha$ be a
$1$-cocycle in $\tri$, that is, a representative of an element in 
$\Hom^1(\tri, \Z_2)$. Following the definition of $1$-cohomology it can be shown
that on each triangle, $\alpha$ evaluates to $1$ on none or two of its edges. 
Thus, by colouring all the edges contained in $\alpha$ by $1/2$ and the 
remaining ones by $0$, $\alpha$ defines an element in $\adm (\tri,3)$. 

\begin{proposition}
  \label{prop:reduction}
  Let $\tri$ be a $3$-manifold triangulation with edge set $E$, $r\geq3$,
  and $\theta \in \adm(\tri,r)$. Then the {\em reduction of $\theta$},
  defined by $\theta' : E \to \{0, 1/2 \}; \,\, e \mapsto \theta(e) - \lfloor 
  \theta(e) \rfloor $, is an admissible colouring in $\adm (\tri, 3)$.
\end{proposition}

\begin{proof}
  Let $f$ be a triangle of $\tri$ with edges $e_1$, $e_2$, and $e_3$.
  Since $\theta \in \adm(\tri,r)$ is admissible, we have $\theta (e_1) + 
  \theta (e_2) + \theta (e_3) \in \mathbb{Z}$. Thus, there are either
  no or two half-integers amongst the colours of the edges of $f$ and
  $\theta' \in \adm(\tri,3)$.  
\qed \end{proof}

Thus every colouring $\theta \in \adm (\tri,r)$ can be associated to a 
$1$-cohomology class of $\tri$ via its reduction $\theta'$. 
We know from \cite{matveev03-algms,turaev92-invariants} that this construction 
can be used to split $\tv_{r,q} (\tri)$ into simpler invariants indexed by
the elements of $\Hom^1(\tri, \Z_2)$. More precisely, let $[\alpha] \in 
\Hom^1(\tri,\Z_2)$ be a cohomology class, then $\tv_r(\tri,[\alpha]) = 
\sum_{\theta \in \adm(\tri,r), \theta' \in [\alpha]} 
|\tri|_\theta $, where $\theta'$ denotes the reduction of $\theta$, is an invariant 
of $\tri$. The special case $\tv_r(\tri,[0])$ is of particular importance for 
computations as explained in further detail in Section~\ref{sec:bounds}.

\subsection{Weights of the Turaev-Viro invariant at $\SL_2(\C)$}
\label{app:tv-weights}

Let $\tri$ be a generalised triangulation of a closed 3-manifold $M$, let $r$ 
and $q$ be integers with $r \geq 3$, $0 < q < 2r$, and let $\gcd(r,q)=1$.
We define the Turaev-Viro invariant $\tv_{r,q}(\tri)$ at $\SL_2(\C)$ as follows.

Let $V$, $E$, $F$ and $T$ denote the set of vertices, edges,
triangles and tetrahedra respectively of the triangulation $\tri$.
Let $I = \{0, 1/2, 1, 3/2, \ldots, (r-2)/2\}$.
For each admissible colouring $\theta$ and
for each vertex $v \in V$, edge $e \in E$, triangle $f \in F$
or tetrahedron $t \in T$, we define \emph{weights}
$|v|_\theta, |e|_\theta, |f|_\theta, |t|_\theta \in \C$. 

\medskip
Our notation differs slightly from Turaev and Viro \cite{turaev92-invariants};
most notably, Turaev and Viro do not consider
triangle weights $|f|_\theta$, but instead incorporate an additional
factor of $|f|_\theta^{1/2}$ into each tetrahedron weight
$|t|_\theta$ and $|t'|_\theta$ for the two tetrahedra $t$ and $t'$ containing $f$.
This choice of notation simplifies the notation and avoids unnecessary
(but harmless) ambiguities when taking square roots.

Let $\zeta = e^{i \pi q / r} \in \C$.  Note that our conditions imply that
$\zeta$ is a $(2r)$-th root of unity, and that $\zeta^2$ is a
\emph{primitive} $r$-th root of unity; that is,
$(\zeta^2)^k \neq 1$ for $k=1,\ldots,r-1$.
For each positive integer $i$, we define
$[i] = (\zeta^i-\zeta^{-i})/(\zeta-\zeta^{-1})$ and,
as a special case, $[0] = 1$.
We next define the ``bracket factorial''
$[i]! = [i]\,[i-1] \ldots [0]$.
Note that $[r] = 0$, and thus $[i]! = 0$ for all $i \geq r$.

We give each vertex constant weight
\begin{equation*}
|v|_\theta = \frac{\left|\zeta-\zeta^{-1}\right|^2}{2r} ,
\end{equation*}
and each edge $e$ of colour $i \in I$ (i.e.,
for which $\theta(e) = i$) 
\begin{equation*}
|e|_\theta = (-1)^{2i} \cdot [2i+1].
\end{equation*}
A triangle $f$ whose three edges have colours $i,j,k \in I$ is assigned the
weight
\[ |f|_\theta = (-1)^{i+j+k} \cdot
    \frac{[i+j-k]! \cdot [i+k-j]! \cdot [j+k-i]!}{[i+j+k+1]!}. \]
Note that the parity condition and triangle inequalities ensure that
the argument inside each bracket factorial is a non-negative integer.

\begin{figure}[t]
    \begin{center}
        \includegraphics[width = .15\textwidth]{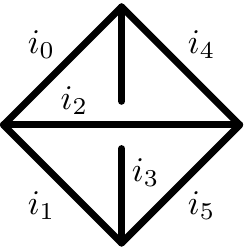}
        \caption{Edge colours of a tetrahedron. \label{fig:tet}}
    \end{center}
\end{figure}

Finally, let $t$ be a tetrahedron with edges colours
$i_0,i_1,i_2,i_3,i_4,i_5$ as indicated in Figure~\ref{fig:tet}. In particular,
the four triangles surrounding $t$ have colours
$(i_0,i_1,i_3)$, $(i_0,i_2,i_4)$, $(i_1,i_2,i_5)$ and $(i_3,i_4,i_5)$,
and the three pairs of opposite edges have colours
$(i_0,i_5)$, $(i_1,i_4)$ and $(i_2,i_3)$.  We define
\begin{align*}
\tau_\theta(t,z) &=
    [z-i_0-i_1-i_3]! \cdot [z-i_0-i_2-i_4]! \cdot
    [z-i_1-i_2-i_5]! \cdot [z-i_3-i_4-i_5]!\,, \\
\kappa_\theta(t,z) &= [i_0+i_1+i_4+i_5-z]! \cdot
                  [i_0+i_2+i_3+i_5-z]! \cdot
                  [i_1+i_2+i_3+i_4-z]!
\end{align*}
for all integers $z$ such that the bracket factorials above all have non-negative
arguments; equivalently, for all integers $z$ in the range
$\zmin \leq z \leq \zmax$ with
\begin{align*}
\zmin &= \max\{i_0+i_1+i_3,\ i_0+i_2+i_4,\ i_1+i_2+i_5,\ i_3+i_4+i_5\}\,; \\
\zmax &= \min\{i_0+i_1+i_4+i_5,\ i_0+i_2+i_3+i_5,\ i_1+i_2+i_3+i_4\}.
\end{align*}
Note that, as before, the parity condition ensures
that the argument inside each bracket factorial above is an integer.
We then declare the weight of tetrahedron $t$ to be
\begin{equation*}
|t|_\theta = \sum_{\zmin \leq z \leq \zmax}
    \frac{(-1)^z \cdot [z+1]!}{\tau_\theta(t,z) \cdot \kappa_\theta(t,z)},
\end{equation*}
%
Note that all weights are polynomials in $\zeta$ with rational coefficients, 
where $\zeta = e^{i \pi q/r}$.
Using these weights, we can define the weight $|\tri|_\theta$ of an edge 
colouring $\theta$ as done in Equation~(\ref{eq:colouring}), 
%
and the Turaev-Viro invariant to be the sum of the weights 
of all admissible colourings
\[ \tv_{r,q}(\tri) = \sum_{\theta\ \mathrm{admissible}} |\tri|_\theta. \]

%
%
%
%

%


\section{Admissible colourings of tetrahedra and embedded surfaces}
\label{sec:col-tetrahedron}


In the definition of Turaev-Viro type invariants, colours are assigned to edges 
and the admissibility conditions only depend on the triangles. In this section,
we translate these conditions into a characterisation of the ``admissible 
colourings'' of a tetrahedron. At the end of the section we discuss the 
connection between tetrahedra colourings and the theory of embedded surfaces in 
$3$-manifolds.

Let $\tri$ be a triangulated $3$-manifold and let $t$ be a tetrahedron. We interpret a colouring of the edges of $t$ as a system of disjoint polygonal cycles on the boundary of $t$ (see Theorem~\ref{thm:normalarcs}). We characterise these cycles in terms of their ``intersections patterns'' with the edges of $t$. To do so, we translate this combinatorial problem into the classification of simple closed curves in the $3$-punctured disk $\disk_3$. We prove that the notion of ``intersection patterns'' is well-defined for isotopy classes of simple closed loops in $\disk_3$. Finally we study the action of the \emph{mapping class group} of $\disk_3$ on these isotopy classes and on their intersections with the edges.
Due to the compact representation of the mapping class group as the braid group, and the symmetries of the tetrahedron, we reduce the classification of simple closed curves to an inductive argument using a small case study.

To motivate this study, we rewrite the formulae for the original 
Turaev-Viro invariant at $\SL_2(\C)$ in this new ``system of coordinates'' (see Theorem~\ref{thm:weights}). The formulae are substantially simpler, making this approach to Turaev-Viro type invariants appear promising.



\medskip
\noindent
{\bf System of polygonal cycles:} We give an interpretation of admissible colourings on a triangulation $\tri$ in terms of {\em normal arcs}, i.e., straight lines in 
the interior of a triangle which are pairwise disjoint and meet the edges of a 
triangle, but not the vertices (see Figure~\ref{fig:normalarcs}).

\begin{figure}[t]
  \centering
    \includegraphics[width = 5.1cm]{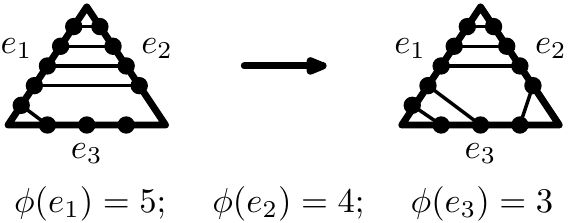}
%
  \caption{Constructing a system of normal arcs from edge colourings by reorganising matchings.} 
  \label{fig:normalarcs}
\end{figure}

For a colouring $\theta(e)$ of an edge $e$, we define 
$\phi(e) = 2\theta(e) \in \Z$, and we use the term ``colouring'' for $\phi$
for the remainder of this section. This way the colours $\phi(e_1),\phi(e_2),
\phi(e_3) \in \{0, 1, \ldots , r-2\}$ can be interpreted as the number of 
intersections of normal arcs with the respective edges of the triangulation 
(see Figure~\ref{fig:normalarcs}). 

For a tetrahedron $t$ with edges $\{e_0, \ldots, e_5\}$, and a colouring of $t$ 
$\phi\colon \{e_0, \ldots, e_5\} \to \{0,1,\ldots, r-2\}$, we define the 
\emph{intersection symbol of $t$} to be the $2\times 3$ matrix of the values $\phi(e_i)$, where the first row contains the colours of the edges of a triangle, and colours of opposite edges appear in the same column; see Figure~\ref{fig:unfold_tet}. 
We treat intersection symbols like matrices, and allow addition and 
multiplication by a scalar. If the two rows of the intersection symbol are 
identical, we write $\intsymb{\phi(e_0)}{\phi(e_1)}{\phi(e_2)}$ for short.
Note how different tetrahedron symmetries act on the entries of an intersection
symbol.


\begin{theorem}[Burton et al.\cite{Burton15TuraevViro}]
    Given a $3$-manifold triangulation $\tri$ and $r \geq 3$,
    an admissible colouring of the edges of $\tri$ with $r-1$ colours 
    corresponds to a system of normal arcs
    in the triangles of $\tri$ with $\leq r -2 $ arcs per 
    triangle forming a collection of \emph{polygonal cycles} on the boundary of 
    each tetrahedron of $\tri$.
\label{thm:normalarcs}
\end{theorem}


\begin{proof}
  Following the definition of an admissible colouring from 
  Section~\ref{ssec:tv}, the colours of the edges $e_1$, $e_2$, $e_3$ of a 
  triangle $f$ of $\tri$ must satisfy the parity condition ($\phi(e_1)+
  \phi(e_2)+\phi(e_3)$ even) and the triangle inequalities.

  Without loss of generality, let $\phi(e_1) \geq \phi(e_2) \geq \phi(e_3)$. 
  We construct a system of normal arcs by first drawing $\phi(e_2)$ arcs between
  edge $e_1$ and $e_3$ and $\phi(e_1) - \phi(e_2)$ arcs between edge $e_1$ and 
  $e_3$. This is always possible since $\phi(e_1) \leq \phi(e_2) + \phi(e_3)$ 
  by the triangle inequality. Furthermore, the parity condition ensures that an 
  even number of unmatched intersections remains which, by construction, all 
  have to be on edge $e_3$. If this number is zero we are done. Otherwise we 
  start replacing normal arcs between $e_1$ and $e_2$ by pairs of normal arcs, 
  one between $e_1$ and $e_3$ and one between $e_2$ and $e_3$ (see 
  Figure~\ref{fig:normalarcs}). In each step, the number of unmatched 
  intersection points decreases by two. By the assumption $\phi(e_2) \geq 
  \phi(e_3)$, this yields a system of normal arcs in $f$ which leaves no 
  intersection on the boundary edges unmatched. This system of normal arcs
  is unique for each admissible triple of colours.
  By the upper bound constraint, we get at most $r-2$ normal arcs on $f$.

  Looking at the boundary of a tetrahedron $t$ of $\tri$ these normal arcs
  form a collection of closed polygonal cycles. To see this, note that each 
  intersection point of a normal arc in a triangle with an edge is part of 
  exactly one normal arc in that triangle and
  that there are exactly two triangles sharing a given edge.
\qed \end{proof}

In the following, we classify these polygonal cycles.

\begin{figure}[t]
\centering
\includegraphics[width=13cm]{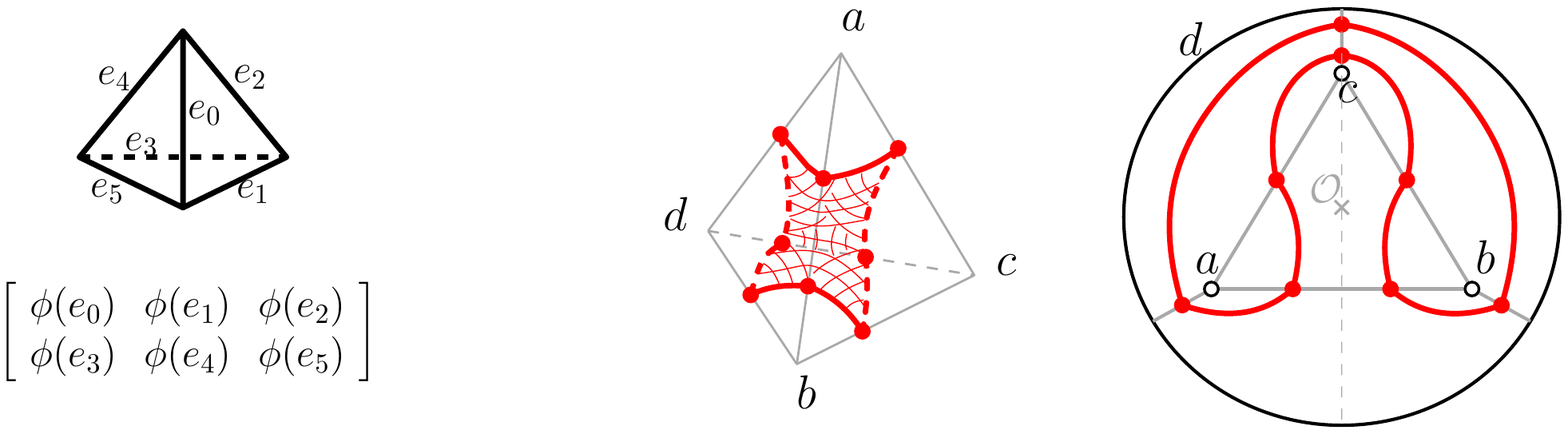}
\caption{Left: Intersection symbol. Right: Tetrahedron with a polygonal cycle formed by $8$ normal arcs, bounding an octagon within the interior of the tetrahedron. Equivalent representation of the cycle as a loop in the disk with punctures $\{a,b,c\}$, and $d$ sent to the boundary. The loop in $\disk_3$ is reduced, and has intersection symbol $\intsymb{2}{1}{1}$.}
\label{fig:unfold_tet}
\end{figure}



\subsection{Topology of the punctured disk} 

A homeomorphism $\X \to \Y$ between two topological spaces is a continuous 
bijective map with continuous inverse. Two topological spaces admitting a 
homeomorphism are said to be \emph{homeomorphic} or \emph{topologically 
equivalent}, and we write $\X \cong \Y$. Two homeomorphisms $f,g \colon \X \to 
\X$ from a \emph{closed} topological space to itself are \emph{isotopic} if 
there exists a continuous map $H: \X \times [0;1] \to \X$ satisfying 
$H(\cdot,0) = f$ and $H(\cdot,1) = g$, and for each $t \in [0;1]$, $H(\cdot, 
t)\colon \X \to \X$ is a homeomorphism. 

Let $\disk$ be the closed $2$-dimensional disk, and let
$a,b,c \in \disk$ be three distinct, arbitrary but fixed, points in its 
interior $\mathring{\disk} = \disk \setminus \partial \disk$. Denote by $\disk_3$ the $3$-punctured disk $\disk \setminus \{ a,b,c\}$. 
A self-homeomorphism $f\colon \disk \setminus \{ a,b,c\} \to \disk \setminus \{ a,b,c\}$ 
is \emph{isotopic to the identity} if its completion $\widetilde{f}\colon \disk \to \disk$ 
is isotopic to $\id_{\disk}$, by an isotopy $H\colon \disk \times [0;1] \to \disk$ that 
fixes points $a$, $b$ and $c$. 


Let $\mathrm{Homeo}(\disk_3, \partial \disk)$ be 
the group, under composition, of (orientation preserving) homeomorphisms 
$\disk_3 \to \disk_3$ that are the identity on the outer boundary 
$\partial \disk$, and let $\mathrm{Homeo}_0(\disk_3, \partial \disk)$ be the 
subgroup of such homeomorphisms that are isotopic to the identity. 
We define the \emph{mapping class group of $\disk_3$ relative to 
$\partial \disk$} to be the group quotient:
\[
\mcg(\disk_3, \partial \disk ) = \mathrm{Homeo}(\disk_3, \partial \disk) / 
\mathrm{Homeo}_0(\disk_3, \partial \disk)
\]
that we denote by $\mcg_3$ for short. It is known that $\mcg_3$ is isomorphic to the braid group $B_{3}$~\cite{birman1975braids}. This is the non-abelian group 
generated by two elements $\sigma_1$ and $\sigma_2$, satisfying $\sigma_1 \sigma_2 
\sigma_1 = \sigma_2 \sigma_1 \sigma_2$. 



A \emph{free loop} 
on $\disk \setminus \{a,b,c\}$ is a continuous embedding of the 
circle $S^1$ into $\mathring{\disk} \setminus \{a,b,c\}$, i.e. 
$L \colon S^1 \to \disk \setminus \{a,b,c\}$. A free loop is \emph{simple} if 
it has no self intersection. Two simple free loops $L_1, L_2$ are 
\emph{isotopic} if there exists a self-homeomorphism of $\disk \setminus 
\{a,b,c\}$ isotopic to the identity that sends the image of $L_1$ to the image 
of $L_2$. Recall that, due to the \emph{Jordan-Schoenflies theorem}~\cite{birman1975braids}, 
a simple 
free loop in the plane separates the plane into two regions, the \emph{inside} and the \emph{outside}, and
there exists a self-homeomorphism of the plane under which the loop is mapped 
onto the unit circle.
We use the term \emph{loop} to denote simple free loops as well as their 
image in $\disk \setminus \{a,b,c\}$ as simple closed curves. Furthermore, we 
assume that all loops are smooth, and cut tetrahedron edges 
transversally. 
We refer to~\cite{birman1975braids} for more details 
about these definitions.

\subsection{Classification of loops by their intersection symbol}

Given a tetrahedron $t$, its boundary $\partial t$ is a topological $2$-sphere. 
Removing each vertex of $t$, seen as a point, leads to the $4$-punctured sphere,
or equivalently the $3$-punctured disk $\disk_3$ (after closing the outer 
boundary). We also embed the tetrahedron edges in $\disk_3$, as illustrated in 
Figure~\ref{fig:unfold_tet}, using straight line segments. We say that a loop in 
$\disk_3$ is \emph{reduced} if it does not cross a tetrahedron edge twice in a 
row. We define the \emph{intersection symbol of a reduced loop} in $\disk_3$ to 
be the $2 \times 3$ integer matrix of intersection numbers of the reduced loop 
with the tetrahedron edges embedded in $\disk_3$. Note that reduced loops are the topological equivalent of the combinatorial ``polygonal cycles'' 
defined in Theorem~\ref{thm:normalarcs} (by convention, we put the crossing numbers of edges $ab$, $bc$ then $bd$ in the first row of intersection symbols). Naturally, the intersection symbol 
of a reduced loop constitutes a valid tetrahedron intersection symbol. 
For a loop $L$ in $\disk_3$, we denote its isotopy class by $[L]$; it is the 
class of all loops isotopic to $L$ in $\disk_3$. We prove that the 
``intersection symbol'' is well-defined for isotopy classes of loops.

\begin{lemma}
\label{lem:reduced_loop}
The following is true:
\begin{enumerate}[(i)]
\item any isotopy class of loops in $\disk_3$ admits a reduced loop,
\item any two isotopic reduced loops have equal intersection symbols,
\item any two non-isotopic reduced loops have distinct intersection symbols.
\end{enumerate}
\end{lemma}


\begin{proof}
(i) Let $L$ be an arbitrary free loop. If $L$ is reduced, then $[L]$ contains 
a reduced free loop. Otherwise, $L$ crosses the same edge twice in a row. In this
case we can deform $L$ locally via an isotopy, reducing the number of 
intersections between $L$ and the tetrahedron edges by two.
Reproducing this deformation eventually leads to a reduced loop $L' \in [L]$. 

(ii) Let $\pi_1(\disk_3,\mathcal{O})$ be the fundamental group of the $3$-punctured disk with 
base point being the center of the triangle $\mathcal{O}$, see 
Figure~\ref{fig:unfold_tet}. It is a classic result in planar topology (see for 
example~\cite{birman1975braids}) that this group is the free non-abelian group
with $3$ generators. Fixing an orientation, each of these generators is the 
homotopy class of the loop going around exactly one of the punctures exactly 
once---with the proper orientation. Equivalently, a generator is a loop that 
passes through exactly one of the segments $ad$, $bd$ and $cd$ in Figure~\ref{fig:unfold_tet} 
once---in the proper direction. Denote these generators by $\ell_a$, $\ell_b$ 
and $\ell_c$. 

Let $L_1$ and $L_2$ be two isotopic, reduced, simple free loops. Fix points
$x_1$ on $L_1$ and $x_2$ on $L_2$. Their intersection patterns with the line
segments in $\disk_3$, read starting from $x_1$ and $x_2$ respectively, define 
two words in $\pi_1(\disk_3) = \langle \ell_a,\ell_b,\ell_c \rangle$, denoted by 
$l_1$ and $l_2$ respectively. It is known (see for example~\cite{hatcher02-algebraic}) 
that, for $L_1$ and $L_2$ isotopic free loops, $l_1$ and $l_2$ must be conjugate, 
i.e. there exists a word $w$ such that $l_1 = w l_2 w^{-1}$. Thus we can choose
a new base-point $x'_1$ on $L_1$ giving rise to $l'_1 =w^{-1} w l_2$, but 
$L_1$ was reduced, and thus $w$ must be empty, $l'_1 = l_2$, and $L_1$ and
$L_2$ must have equal intersection symbols.

(iii) Suppose that two reduced loops $L_1, L_2$ have same intersection symbol
$s$. Using the construction from Theorem~\ref{thm:normalarcs}, we draw a 
``canonicial reduced loop'' $L$ for the admissible symbol $s$, by fixing points on 
tetrahedron edges for each intersection described in $s$, and drawing the unique
system of normal arcs to get $L \subset  \disk_3$. Because $L_1$ and $L_2$ are 
reduced, the restriction of $L_1$ (or $L_2$) to any triangular face 
(defined by the tetrahedron edges) is isotopic to the restriction of $L$ on this
face. Since the intersection points on the triangular boundaries have to align,
the isotopy can be made global, and both $L_1$ and $L_2$ are isotopic to $L$, 
hence $L_1$ and $L_2$ are isotopic.
\qed \end{proof}

It follows that we can refer to the intersection symbol of an isotopy class of 
loops $[L]$ as the intersection symbol of any reduced loop in $[L]$. By a small 
abuse of notation, we also refer to the intersection symbol of a loop $L$ as the
intersection symbol of $[L]$.

By virtue of the Jordan-Schoenflies theorem, we distinguish 
three types of loops: (i) loops containing no puncture in the inside; 
(ii) loops separating one puncture from the three others; and (iii) loops 
separating two punctures from the two others. Note that here we call the outer 
boundary of $\disk_3$ ``puncture'' as well. Naturally, loops of type (i) are 
trivial and have intersection symbol $\intsymb{0}{0}{0}$, and loops of type 
(ii) can be isotoped to a circle in a small neighbourhood of the puncture in 
their inside, and hence have $(2\times 3)$-intersection symbol 
$\left [ \begin{array}{lll} 1 & 0 & 1 \\ 0 & 1 & 0 \end{array} \right ]$, up to tetrahedron permutations. 
The case of loops of type (iii) is more interesting; we call these loops 
\emph{balanced}. We prove:

\begin{lemma}
\label{lem:loop_homeo}
For any two loops $L_1$ and $L_2$ of same type (i), (ii) or (iii), there exists a 
homeomorphism of $\disk_3$, constant on $\partial \disk$, sending $L_1$ to 
$L_2$.
\end{lemma}

\begin{proof}
This is a consequence of the Jordan-Schoenflies 
theorem. Consider the completion of $\disk_3 = \disk \setminus \{a,b,c\}$ into 
$\disk$ by filling up the three punctures. Let $h_1\colon \disk \to \disk$ and 
$h_2\colon \disk \to \disk$ be two self-homeomorphisms of the plane sending 
$L_1$ and $L_2$, respectively, to the unit circle. Consequently, $h_2^{-1} \circ 
h_1$ is a self-homeomorphism of $\disk$ sending $L_1$ to $L_2$. Since $L_1$ 
and $L_2$ are of the same type, their inside and outside in 
$\disk \setminus \{a,b,c\}$ are homeomorphic, by a homeomorphism that preserves 
the boundary $L_2 \cup \partial \disk$ (this homeomorphism ``aligns'' punctures). 
Composing $h_2^{-1} \circ h_1$ with this homeomorphism sends $\{a,b,c\}$ to 
$\{a,b,c\}$ in $\disk$, and defines the self-homeomorphism of $\disk_3$ sending 
$L_1$ to $L_2$.
\qed \end{proof}

\begin{figure}[t]
\centering
\includegraphics[width=6.5cm]{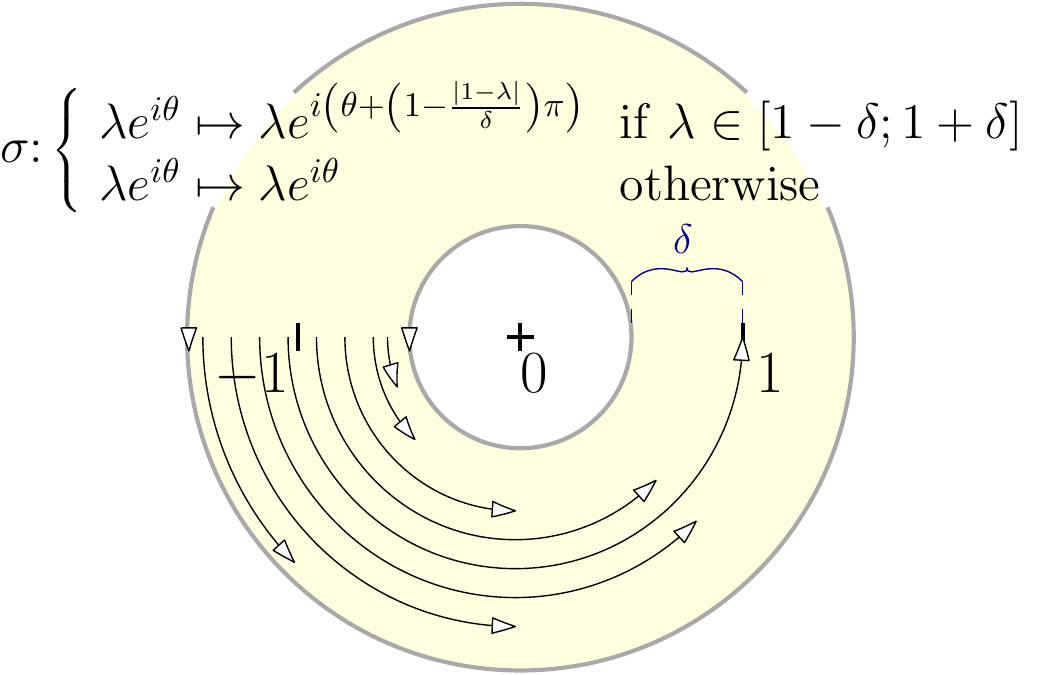} \hspace{0.5cm}
\includegraphics[width=6.5cm]{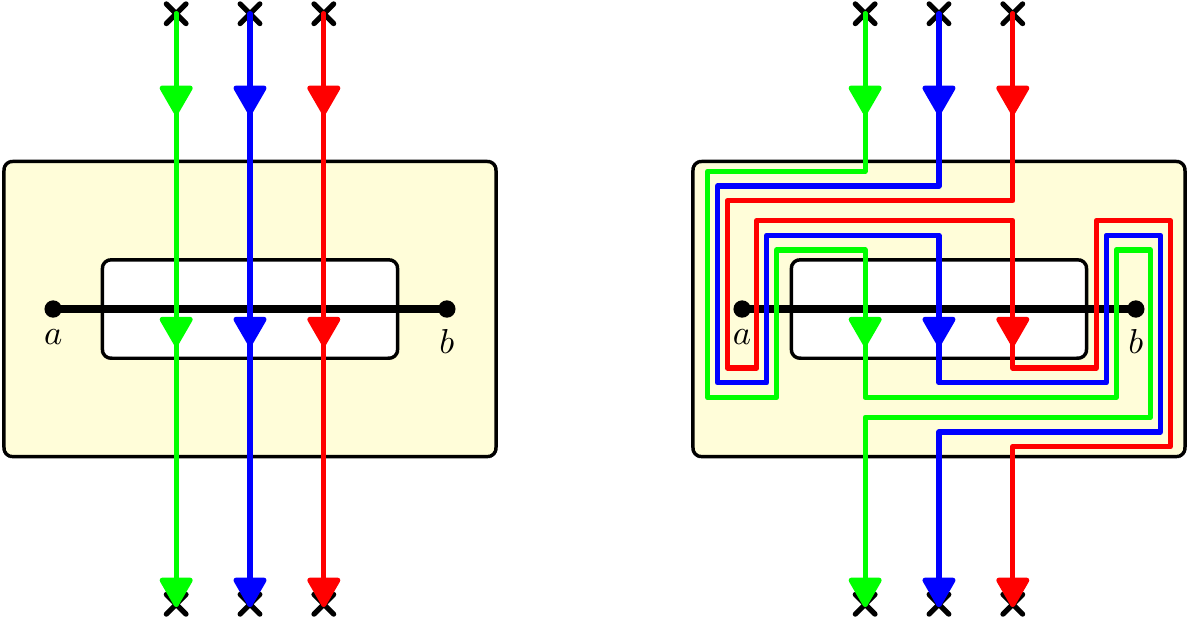}
\caption{Left: Construction of a homeomorphism in the complex plane exchanging 
$-1$ with $1$ (in the positive direction) as a model to exchange positions of 
punctures $a$ and $b$ in $\disk_3$. Right: Action of the homeomorphism exchanging punctures at $a$ and $b$ on pieces of loop crossing edge $ab$ transversally (the annulus has been placed in a close neighbourhood of edge $ab$, with $-1$ on puncture $a$ and $1$ on puncture $b$). Note that the curves must cross the annulus and edge $ab$ transversally; to maintain this property, we assume that we deform isotopically the loops ``outside'' the annulus after applying the homeomorphism.}
\label{fig:sigma_mcg}
\end{figure}

As a consequence, intersection symbols of loops are defined up to isotopy with the identity, and any pair of loops are related by a homeomorphism. Hence, the action of an element (i.e. a class of homeomorphisms) of the mapping class group $\mcg_3$ on an intersection symbol is well defined. 
Before classifying balanced loops, we give an explicit characterisation of the 
generators of the mapping class group $\mcg_3$, coming from the 
isomorphism with the braid group $B_3$. These generators are classes of 
homeomorphisms, exchanging two punctures. See Figure~\ref{fig:sigma_mcg} for an 
explicit homeomorphism exchanging punctures $a$ and $b$, and its local action on
curves intersecting the line segment $ab$ transversally. The homeomorphism is the 
identity everywhere except for in a small annulus containing $a$ and $b$. 
Denote by $\sigma_{ab}$ and $\sigma_{bc}$ the homeomorphisms, as 
defined in Figure~\ref{fig:sigma_mcg}, exchanging punctures $a$ with $b$, and 
punctures $b$ with $c$ respectively. We now classify the intersection symbols of balanced loops.

\begin{theorem}
  There is a bijection between isotopy classes of balanced loops in $\disk_3$ and 
  intersection symbols of the form $\intsymb{i}{j}{i+j}$, up to tetrahedron 
  permutation, with $i$, $j$ coprime non-negative integers.
  \label{thm:classificationsymb}
\end{theorem}

\begin{figure}[t]
\centering
\includegraphics[width=12cm]{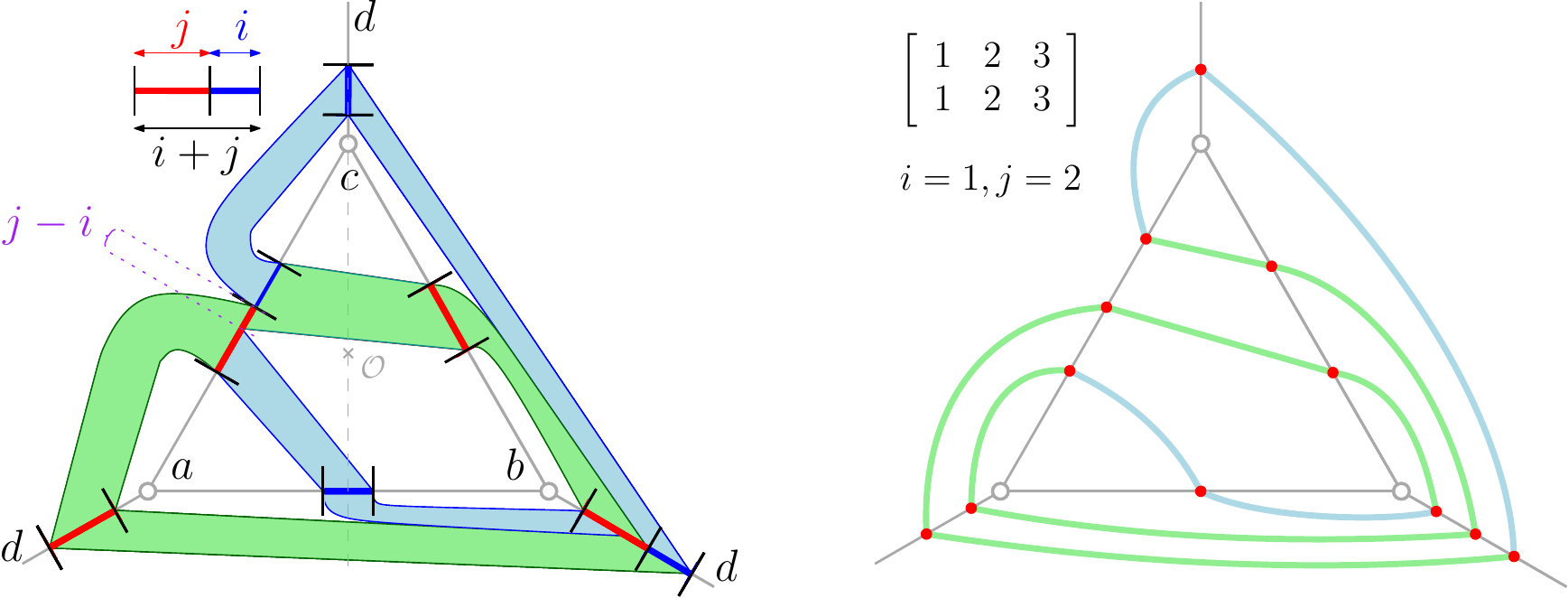}
\caption{Reduced loop, given by the intersection symbol $\intsymb{i}{j}{i+j}$ 
with $j \geq i$. Blue domains represent $i$ parallel segments and green domains 
represent $j$ parallel segments. Note in particular that the $j$ segments, in 
the green domain inside the central triangle, originating from edge $bc$ and 
crossing edge $ac$, split into $j-i$ segments crossing $ad$ and $i$ segments 
crossing $cd$. The center $\mathcal{O}$ and vertical axis are drawn in grey.}
\label{fig:planar_tet}
\end{figure}

\begin{proof}
Recall that we denote intersection symbols with two identical rows only by 
their first row $\intsymb{i}{j}{i+j}$ and, by symmetry, we suppose $j \geq i$.
The proof is separated into two parts (i) (intersection symbol equals
$\intsymb{i}{j}{i+j}$) and (ii) ($i$ and $j$ are coprime). 

(i) We first prove that balanced loops have intersection symbols 
$\intsymb{i}{j}{i+j}$, up to permutation, for \emph{arbitrary} $i,j \geq 0$. By 
virtue of Lemma~\ref{lem:loop_homeo}, any balanced loop $L$ in $\disk_3$ may be 
obtained via a homeomorphism $h\colon \disk_3 \to \disk_3$ from the reduced 
loop $L_0$ with intersection symbol $\intsymb{0}{1}{1}$. By virtue of the isomorphism $\mcg_3 \cong B_3$, this homeomorphism $h$ 
can be expressed as $g \circ f$, where $f$ is a composition of the 
homeomorphisms $\sigma_{ab}$ and $\sigma_{bc}$, and $g$ is isotopic to the 
identity. Because intersection symbols are defined up to isotopy with 
the identity, we focus on the action of the homeomorphisms $\sigma_{ab}$ and 
$\sigma_{bc}$ on the intersection symbol of an (balanced) loop. 
Recall that $\sigma_{ab}$ and $\sigma_{bc}$ are homeomorphisms that exchange punctures $a$ with $b$, and puncture $b$ with $c$ respectively, in the positive direction. We specify the homeomorphism $\sigma_{ab}$ exactly in Figure~\ref{fig:sigma_mcg} and show its action on curves intersecting edge $ab$ transversally. Note that these homeomorphisms are similar to \emph{Dehn twists} in the study of surface topology.

We prove the result 
inductively. The intersection symbol of the base case $L_0$ satisfies the 
property with $i=0$ and $j=1$. Suppose that we are given a reduced loop $L$ in 
$\disk_3$ with intersection symbol $\intsymb{i}{j}{i+j}$. 
Figure~\ref{fig:planar_tet} represents such loop. We study the action of 
$\sigma_{ab}$ and $\sigma_{bc}$ on the intersection symbol of such a loop. 

All tetrahedron permutations may be obtained by rotations of $2 \pi / 3$ around 
center $\mathcal{O}$ and reflections along the vertical axis passing through 
puncture $c$; see Figure~\ref{fig:planar_tet}. In order to reduce the study to 
a limited number of cases, note that exchanging any two punctures is equivalent 
to applying the appropriate rotation, exchanging the bottom two punctures and 
rotating back to the original position. Hence we study only the action of 
$\sigma_{ab}$ and $\sigma_{ab}^{-1}$. Additionally, applying $\sigma_{ab}$ to a 
loop $L$ or $\sigma_{ab}^{-1}$ to its reflection $L'$ along the vertical axis, 
is equivalent; specifically, denoting by $X$ a configuration and by 
$\overline{X}$ its reflection along the vertical axis, we have 
$\sigma_{ab}(\overline{X}) = \overline{\sigma_{ab}^{-1}(X)}$. For simplicity
we denote $\sigma_{ab}$ by $\sigma$ for the remainder of this proof.

Figure~\ref{fig:loop_ex} shows a detailed example of the action of $\sigma$ on 
a piece of loop traversing $ab $ transversally. Figure~\ref{fig:loops} pictures 
all possible cases of a piece of loop intersecting edge $ab$, together with a 
neighbourhood of the intersection. These configurations and their reflected 
versions appear on one of the edges $\{ab,bc,ac\}$ in Figure~\ref{fig:planar_tet}, picturing the the loop with intersection symbol $\intsymb{i}{j}{i+j}$ of the inductive argument. 
The action 
of $\sigma$ and $\sigma^{-1}$ is local in the sense that only a piece of the 
loop in a neighbourhood of the intersection with the edge $ab$ needs to be 
considered to reduce the loop after transformation. 
Additionally, the five configurations $A, O_1, O_2, O_3$ and $O_4$ of Figure~\ref{fig:loops} can be regarded independently, because the considered neighbourhoods do not overlap. Finally, the modification of the intersection symbol induced by the homeomorphisms $\sigma$ and $\sigma^{-1}$ on a single piece of loop is pictured by $2 \times 3$-matrices in Figure~\ref{fig:loops}.


\begin{figure}[t]
\centering
\includegraphics[width=13.8cm]{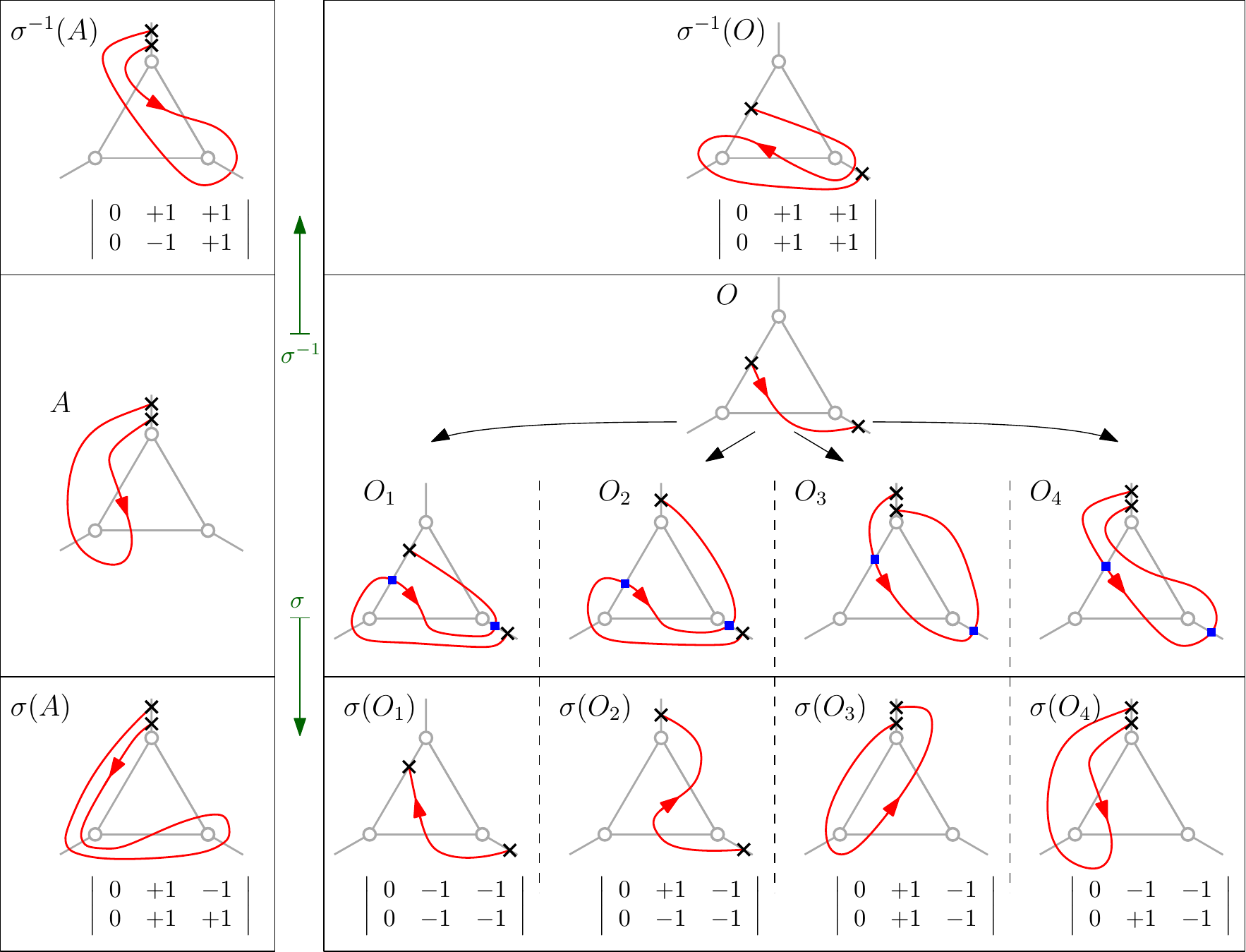}
\caption{Pieces of loop traversing edge $ab$. We distinguish two cases: case 
$A$ where the loop originates from edge $ac$, crosses $ab$ then $ad$, and case 
$O$ where the loop originates from edge $ac$, crosses $ab$ then $bd$. For the 
action of $\sigma$, we distinguish $4$ subcases $O_1$ to $O_4$ to case $O$, as 
the loop gets shortened by isotopy in order to maintain the property of being 
reduced (the original piece $O$ is delimited by blue squares in special cases 
$O_1$ to $O_4$). The black crosses represent fixed points for the action of 
$\sigma$, $\sigma^{-1}$ and the isotopy reduction; in particular, the loop is 
reduced when the piece between the black crosses does not cross twice the same 
edge, and the piece of loop enters or leave a fixed point with the same 
orientation. The loop pieces $\{A,O_1,O_2,O_3,O_4\}$ and their reflections cover
all possible crossing configurations found on an edge $\{ab,bc,ac\}$ of a loop with 
intersection symbol $\intsymb{i}{j}{i+j}$ as in Figure~\ref{fig:planar_tet}. 
}
\label{fig:loops}
\end{figure}

We apply the homeomorphisms $\sigma$ and $\sigma^{-1}$ to all three $2 \pi /3$ 
rotations of the intersection symbol $\intsymb{i}{j}{i+j}$ and express the 
transformation in terms of the matrices presented in Figure~\ref{fig:loops}. Due to the independence of configurations $A,O_1,O_2,O_3,O_4$, we can linearly sum the $2\times 3$ matrices for all pieces of loop intersecting edge $ab$. 
Below we list all crossing configurations with edge $ab$ for all permutations
of the intersection symbol; the transformation of the intersection symbol can
be calculated in the following way: 
\[
\begin{array}{lcll}
\intsymb{i}{j}{i+j} &\equiv& i \times O_1 & 
    \left\{ \begin{array}{cl}
              \xmapsto{\ \ \sigma \ \ } & \intsymb{i}{j-i}{j} \\
              \xmapsto{\  \sigma^{-1} } & \intsymb{i}{i+j}{2i+j} \\
            \end{array}
    \right. \\
&&&\\
\intsymb{i+j}{i}{j} &\equiv& i \times (A + \overline{A}) + (j-i) \times O_3 & 
    \left\{ \begin{array}{cl}
              \xmapsto{\ \ \sigma \ \ } & \intsymb{i+j}{2i+j}{i} \\
              \xmapsto{\ \sigma^{-1} } & \intsymb{i+j}{j}{i+2j} \\
            \end{array}
    \right. \\
&&&\\
\intsymb{j}{i+j}{i} &\equiv& \begin{array}{ll}
                    \max(0,2i-j) \times \overline{O_1} + \max(0,j-2i) \times 
                    \overline{O_3} \\ 
                    + \min(i,j-i) \times (\overline{O_2} + \overline{O_4}) \\
                    \end{array} & 
    \left\{ \begin{array}{cl}
              \xmapsto{\ \ \sigma \ \ } & \intsymb{j}{i+2j}{i+j} \\
              \xmapsto{\ \sigma^{-1} } & \intsymb{j}{i}{j-i} \\
            \end{array}
    \right. \\
\end{array}
\]
Note that the more intricate case study ($j > 2i$ or not) of the case 
$\intsymb{j}{i+j}{i}$ is due to pieces of loop intersecting edges with a 
``split'' twice in the neighbourhood considered, explaining the $2^2 = 4$ subcases 
(see the ``split'' intersection patterns on
edges $bd$ and $ac$ in Figure~\ref{fig:planar_tet}).
In summary, applying $\sigma$ preserves the fact that the intersection symbol is
of the form $\intsymb{i}{j}{i+j}$, up to permutation.

(ii) We now prove that $i$ and $j$ are coprime integers. Note that we can 
simulate the Euclidean algorithm on the pair $(j,i)$ by applying $\sigma$. 
W.l.o.g., consider the intersection symbol $\intsymb{i}{j}{i+j}$, $j \geq i$. 
Note that, if $j > 1$ then $j > i$: assuming otherwise $i=j>1$, then
applying $\sigma$ induces the symbol $\intsymb{i}{0}{i}$ which represents 
multiple parallel disjoint copies of the loop 
$\intsymb{1}{0}{1}$. Consider the division $j = q i + r$, with $r < i$. 
Applying $\sigma$ $q$ times to $\intsymb{i}{j}{i+j}$ gives $\intsymb{i}{r}{i+r}$
with $r < i$. Hence we can recursively apply the Euclidean algorithm on the pair
$(i,r)$. Because of the property that $j \neq i$, the algorithm terminates on 
the pair $(1,0)$ and $\gcd(i,j)=1$. Conversely, for any coprime integers 
$(p,q)$, running the Euclidean algorithm in reverse gives us the sequence of 
homeomorphisms $\sigma_{ab}$, $\sigma_{bc}$, and their inverses to obtain 
the loop $\intsymb{p}{q}{p+q}$ from the loop $\intsymb{0}{1}{1}$.

The bijection now follows by virtue of Lemma~\ref{lem:reduced_loop}. 
\qed \end{proof}

\begin{figure}[t]
\centering
\includegraphics[width=10cm]{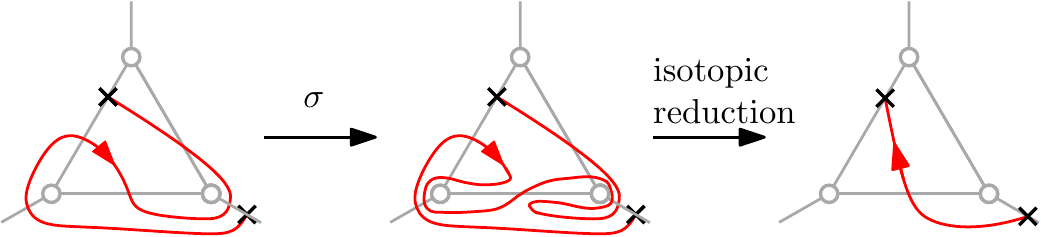}
\caption{Action of $\sigma_{ab}$, with isotopic reduction, on a piece of loop and its neighbourhood (we call this loop configuration $O_1$).}
\label{fig:loop_ex}
\end{figure}


\begin{theorem}
\label{thm:poly}
An admissible tetrahedron colouring corresponds to a family of polygonal curves 
with intersection symbols
\[
a \times \left[ \begin{array}{ccc} 1 & 0 & 1 \\ 0 & 1 & 0 \\ \end{array} \right] +
b \times \left[ \begin{array}{ccc} 1 & 1 & 0 \\ 0 & 0 & 1 \\ \end{array} \right] +
c \times \left[ \begin{array}{ccc} 0 & 1 & 1 \\ 1 & 0 & 0 \\ \end{array} \right] +
d \times \left[ \begin{array}{ccc} 0 & 0 & 0 \\ 1 & 1 & 1 \\ \end{array} \right]
\]
for arbitrary integers $a,b,c,d \geq 0$, and $p \geq 0$ copies of the {\bf same}
polygonal curve with intersection symbol
\[
\left[ \begin{array}{ccc} i & j & i+j \\ i & j & i+j \\ \end{array} \right] \ \text{or} \
\left[ \begin{array}{ccc} i+j & i & j \\ i+j & i & j \\ \end{array} \right] \ \text{or} \
\left[ \begin{array}{ccc} j & i+j & i \\ j & i+j & i \\ \end{array} \right]
\] 
for arbitrary coprime integers $i,j \geq 0$, not both $0$. Integers $a,b,c,d,p,
i,j$ are defined such that the bound $r-2$ on edge colourings is respected.
\label{thm:mainresult}
\end{theorem}


\begin{proof}
The seven intersection symbols are exactly the ones of the isotopy classes of 
loops in $\disk_3$, and hence exactly the ones of all possible polygonal curves 
on the boundary of a tetrahedron (Theorem~\ref{thm:classificationsymb}). 
In $\disk_3$, we can draw an arbitrary number 
of loops separating one puncture from the three others, by drawing them in a 
close neighbourhood of the puncture they isolate. These loops are exactly the 
ones with one of the four first intersection symbols of the theorem.

We prove that a tetrahedron $t$ can have only one type of polygonal cycle with 
one of the three last intersection symbols of the theorem. Take two such 
polygonal loops; by definition, they are disjoint on $\partial t$. Let $L_1$ 
and $L_2$ be the corresponding two pairwise disjoint loops in $\disk_3$; they 
are balanced by definition. Because they are disjoint, suppose, w.l.o.g., that $L_2$ is 
contained in the inside of $L_1$. As $L_1$ and $L_2$ both contain two 
punctures in their inside, the ``band'' between $L_1$ and $L_2$ (i.e. the 
inside of $L_1$ minus the inside of $L_2$) contains no puncture, and is 
then a topological annulus, with boundary $L_1 \sqcup L_2$. Sliding $L_2$ along 
the annulus gives an isotopy between $L_1$ and $L_2$ that is constant outside 
(a neighbourhood of) the annulus, and in particular fixes the punctures. 
Consequently, $L_1$ and $L_2$ have same intersection symbol (Lemma~\ref{lem:reduced_loop}(ii)).
Conversely, given an balanced loop $L$, an arbitrary number of loops with 
same intersection symbol can be drawn in a neighbourhood of $L$.
\qed \end{proof}

In conclusion, Theorem~\ref{thm:mainresult} gives an explicit characterisation 
of admissible tetrahedron colourings in terms of polygonal cycles.
We use this ``system of coordinates'' to re-write the formulae
for the weights of the tetrahedra as defined in 
Section~\ref{app:tv-weights} for the Turaev-Viro invariant at $\SL_2(\C)$. 
Namely, we have the following observation.
%
%
%
%
%

\begin{theorem}
  \label{thm:weights}
  Let $t = \{a,b,c,d\}$ be a tetrahedron, coloured as in 
  Theorem~\ref{thm:mainresult}, i.e. with $a$ (respectively $b$, $c$ and $d$) 
  copies of a $3$-cycle around vertex $a$ (respectively, around vertices $b$, 
  $c$ and $d$), and $p$ copies of the balanced loop $\intsymb{i}{j}{i+j}$. 
  Then, we can express the weight of $t$ as
  \[
  |t| = (-1)^X \sum_{0 \leq z \leq d} \frac{(-1)^z [X - z +1]!}{[a-z]![b-z]![c-z]![d-z]![pi+z]![pj+z]![z]!},
  \]
  where $X := p(i+j)+a+b+c+d$ and $y := \min \{a,b,c,d\}$, and, if $y=0$,
  \[
  |t| = \frac{[X +1]!}{[a]!\ [b]!\ [c]!\ [d]!\ [pi]!\ [pj]!}.
  \]
\end{theorem}
 

\begin{proof}
  The proof consists of an explicit calculation. 
  Given a coloured tetrahedron $t$ with intersection symbol as in 
  Figure~\ref{fig:colour_formulae}. Denote its set of colours by $\Phi$. 
  Note that the ``edge colours'' on the picture are integers---as in the 
  definition of the intersection symbol---and must be divided by two to fit 
  the definition of the Turaev-Viro invariant in Section~\ref{app:tv-weights}, 
  using half-integers.

\begin{figure}[t]
\centering
\includegraphics[width=10cm]{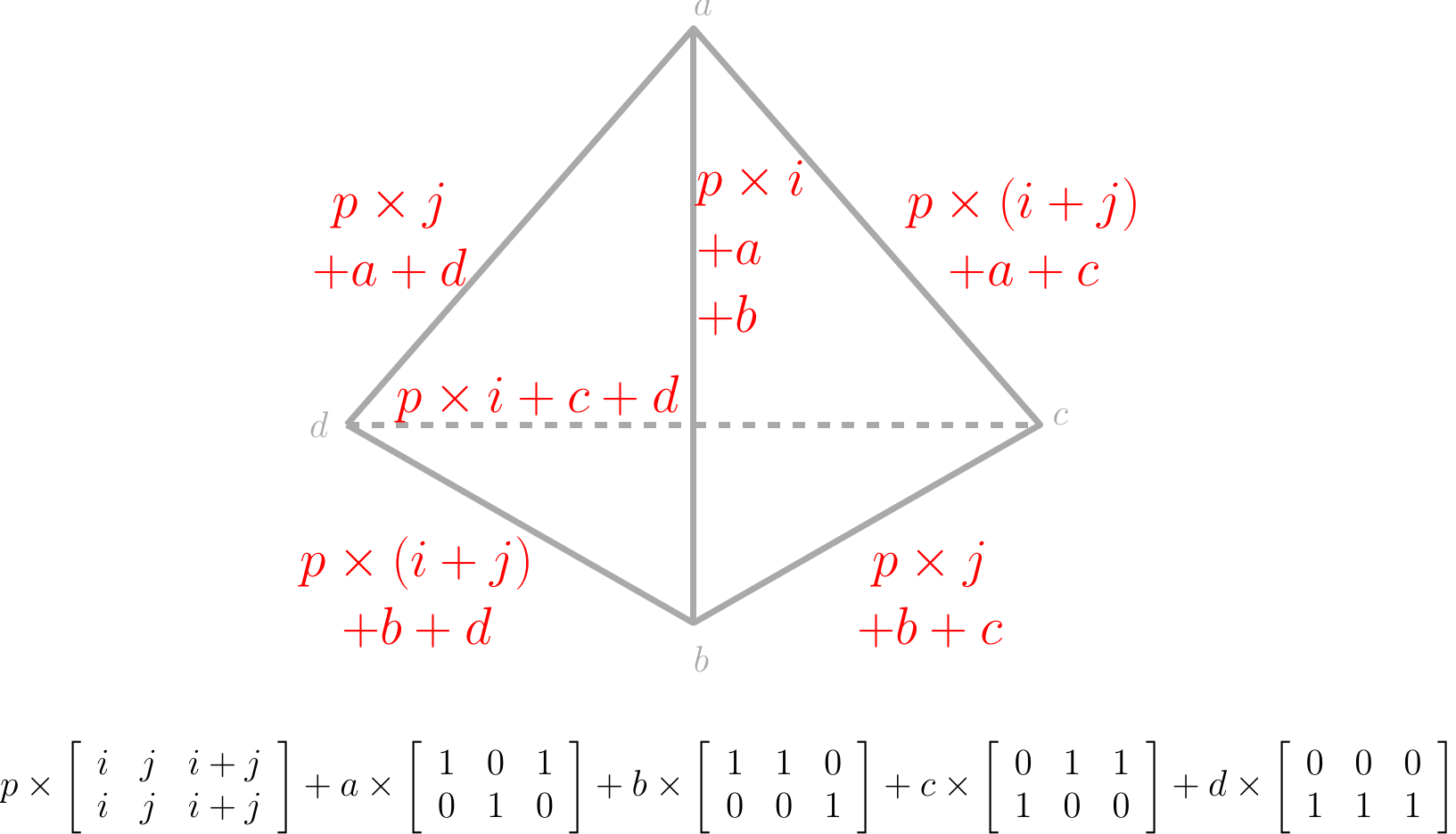}
\caption{Colouring a tetrahedron using the coordinate system of Theorem~\ref{thm:mainresult}. }
\label{fig:colour_formulae}
\end{figure}

  In Section~\ref{app:tv-weights}, $z^-$ and $z^+$ are defined as the maximum 
  values of the sum of edge colours of a triangular face, and the minimum values
  of the sum of the edge colours of a quad (i.e., all edges but two opposite 
  ones). Summing the colours, divided by two yields
  \[
\text{Triangles}: \ \left\{ 
  \begin{array}{l}
    p \times (i+j) + a+b+c\\
    p \times (i+j) + a+b+d\\
    p \times (i+j) + a+c+d\\
    p \times (i+j) + b+c+d\\
  \end{array}
\right.
\ \ \text{Quads}: \ \left\{ 
  \begin{array}{l}
    p \times (i+j) + p\times j + a+b+c+d \\
    p \times (i+j) + p\times i + a+b+c+d \\
    p \times (i+j) + a+b+c+d .\\
  \end{array}
\right.
  \]
Let $y := \min \{a,b,c,d\}$. Consequently, $z^- = p \times (i+j) + a+b+c+d - y$ and $z^+ = p \times (i+j) + a+b+c+d$. Replacing the variable $z$, in the definition of a tetrahedron weight in Section~\ref{app:tv-weights}, by $z- z^-$ we obtain
\[
|t|_{\Phi} = \sum_{0 \leq z \leq y}  \frac{(-1)^{p \times (i+j) + a+b+c+d } \times (-1)^{z-y}  [z+ (p \times (i+j) + a+b+c+d -y) +1]!}{[z+a-y]![z+b-y]![z+c-y]![z+d-y]!  \times  [y+pi-z]![y+pj-z]![y-z]!} ;
\]
and introducing $X := p \times (i+j) + a+b+c+d $, and substituting $z$ by $y-z$ results in
\[
|t|_{\Phi} = (-1)^{X} \sum_{0 \leq z \leq y}  \frac{ (-1)^{z}  [X - z +1]!}{[a-z]![b-z]![c-z]![d-z]!  \times  [pi+z]![pj+z]![z]!} .
\]
\qed \end{proof}

\subsection{Admissible colourings as embedded surfaces}
In this section, we give a more intuitive and 
topological understanding of admissible colourings in terms of embedded surfaces within the triangulation. By interpreting the polygonal cycles as intersection patterns of these embedded surfaces with the boundary $\partial t$ of the tetrahedron $t$, an admissible colouring can be seen as a family of embedded surfaces.

This approach is a powerful tool in computational topology of $3$-manifolds. It is in particular of key importance in the unknot recognition algorithm~\cite{hass99-knotnp}---using \emph{normal surfaces}---and in the $3$-sphere recognition algorithm~\cite{rubinstein97-3sphere}---using \emph{almost normal surfaces}. Normal surfaces consider embedded surfaces cutting through $\partial t$ with $3$-gons, and $4$-gons $\intsymb{0}{1}{1}$. Almost normal surfaces allow, in addition, $8$-gons $\intsymb{1}{1}{2}$. Theorem~\ref{thm:poly} states that Turaev-Viro invariants, for $r$ sufficiently large, consist of formulae involving weights defined on much more complicated surface pieces, with intersection symbol $\intsymb{i}{j}{i+j}$. These intersections are ``helicoidal'' surface pieces of {\em higher index}; see \cite{Bachmann12Helical} for a recent study on embedded surfaces containing these pieces. 

The efficient algorithm for computing the Turaev-Viro invariant for $r=3$ is based on a relation between $\tv_3$ and $2$-homology, and can be interpreted in terms of embedded surfaces. A generalisation of this idea to design more efficient algorithms for arbitrary $r>3$, using the classification of embedded surface pieces from this section, is subject of ongoing research.

\section{Bounds on the number of admissible colourings}
\label{sec:bounds}

Given a $3$-manifold triangulation $\tri$ with $v$ vertices, $n+v$ edges, $2n$ triangles and $n$ tetrahedra (the relations between number of faces follows from an Euler characteristic argument), 
and following the definitions in Section~\ref{ssec:tv} 
above, there are at most $(r-1)^{n+v}$ admissible colourings for $\tv_{r,q}$.
This bound is usually far from being sharp. However, current enumeration
algorithms for admissible colourings do not try to capitalise on this fact
in a controlled fashion. 

In this section we discuss improved upper bounds on the number of admissible 
colourings in important special cases. Moreover, we give a number of 
examples where these improved bounds are actually attained. The bounds are
used in Section~\ref{sec:algo} to construct a significant exponential speed-up
for the computation of $\tv_{r,1}$ where $r$ is odd. 

Note that in the following, we go back to the convention of using half-integers for the colourings of edges.

\subsection{Vertices and the first Betti number}

First let us have a look at some triangulations where we can a priori expect
a rather large number of colourings.

\begin{proposition}
  \label{prop:betav}
  Let $\tri$ be a $3$-manifold triangulation with $v$ vertices. Then
  $$ | \adm(\tri,3) | = 2^{v+\beta_1(\tri,\Z_2)-1}. $$ 
\end{proposition}

\begin{proof}
  Every colouring $\theta \in \adm(\tri,3)$ can be associated to a $1$-cocycle 
  $c_{\theta}$ of $\tri$ over the field with two elements $\Z_2$: Simply define 
  $c$ to be the $1$-cocycle evaluating to $c(e) = 1$ on edges coloured 
  $\theta(e) = 1/2$, and to $c(e) = 0$ on edges coloured $\theta(e) = 0$.
  Conversely, every $1$-cocycle $c$ defines an admissible 
  colouring $\theta_c \in \adm(\tri,3)$. By construction two admissible 
  colourings are distinct if and only if their associated $1$-cocycles are 
  distinct. Hence, the number of admissible colourings $| \adm(\tri,3) |$
  must equal the number of $1$-cocycles of $\tri$.

  First of all $\tri$ has $2^{\beta_1(\tri,\Z_2)}$ $1$-cohomology classes.
  Moreover for every $1$-cocycle $c$ we can find a homologous---but distinct---$1$-cocycle 
  $c'$ by adding a non-zero $1$-coboundary to $c$. The statement now
  follows by noting that the number of $1$-coboundaries of any triangulation 
  $\tri$ equals the number of subsets of vertices of even cardinality, which is 
  $2^{v-1}$.
\qed \end{proof}

\begin{table}[t]
  \begin{center}
    \begin{tabular}{|r|r|r|}
      \hline
      $n$&$\#$ trigs.&$\#$ sharp~(\ref{eq:short}) \\
      \hline
      \hline
      $1$&$1$&$1$ \\
      \hline
      $2$&$4$&$1$ \\
      \hline
      $3$&$24$&$4$ \\
      \hline
      $4$&$160$&$4$ \\
      \hline
      $5$&$1492$&$14$ \\
      \hline
      $6$&$16731$&$22$ \\
      \hline
    \end{tabular}
    \medskip
    \caption{Number ``$\#$ trigs.'' of $1$-vertex triangulations $\tri$ of manifolds with 
      $\beta_1 (\tri,\Z_2) = 1$ and $n$ tetrahedra, $1 \leq n \leq 6$, and
      number of cases of equality ``$\#$ sharp~(\ref{eq:short})'' in the bound from 
      Theorem~\ref{prop:req4}.}
    \label{tab:small}
  \end{center}
\end{table}


Proposition~\ref{prop:betav} is a basic but very useful observation with 
consequences for $\adm(\tri,4)$. This is particularly exiting as 
computing $\tv_{4,1}$ is known to be $\#P$-hard. More precisely, the following
statement holds.

\begin{theorem}
  \label{prop:req4}
  Let $\tri$ be an $n$-tetrahedron $3$-manifold triangulation with $v$ vertices, 
  and let $\theta \in \adm(\tri,3)$. Furthermore, let $\ker_{\theta}$ be the 
  number of edges coloured $0$ by $\theta$. Then
  \begin{align}
    | \adm(\tri,4) | &\leq \left ( \Sigma_{ \theta \in \adm(\tri,3) \setminus \{{\bf 0}\}} 
      2^{\ker_{\theta}} \right )+ 2^{v+\beta_1(\tri,\Z_2)-1} \label{eq:long}  \\
    &\leq (| \adm(\tri,3) |-1) (2^{n+v-1} + 1) + 1, \label{eq:short} 
  \end{align}
  where $\mathbf{0}$ denotes the zero colouring. Moreover, both bounds are sharp.
\end{theorem}

\begin{proof}
  Let $\theta \in \adm (\tri,4)$, and let $\theta'$ be its reduction, as defined in 
  Proposition~\ref{prop:reduction}. If
  $\theta'$ is the trivial colouring (that is, if no colour of $\theta$ is 
  coloured by $1/2$) the colouring $\theta / 2$, obtained 
  by dividing all of the colours of $\theta$ by two, must be in 
  $\adm (\tri, 3)$. It follows from Proposition~\ref{prop:betav} that
  exactly $2^{v+\beta_1(\tri,\Z_2)-1}$ colourings in $\adm (\tri,4)$ reduce
  to the trivial colouring.

  If $\theta'$ is not the trivial colouring then $\theta$ colours some edges
  by $1/2$. In particular it is not the trivial colouring. Since the only 
  colours in $\theta$ are $0$, $1/2$, and $1$, all edges coloured by $1/2$ 
  in $\theta$ are coloured by $1/2$ in $\theta'$ and vice versa. Thus,
  $\ker_{\theta'}$ denotes all edges coloured by $0$ or $1$ in $\theta$.
  Naturally, there are at most $2^{\ker_{\theta'}}$ such colourings.
  The result now follows by adding these upper bounds $2^{\ker_{\theta'}}$
  over all non-trivial reductions $\theta' \in \adm (\tri,3)$, and adding
  the $2^{v+\beta_1(\tri,\Z_2)-1}$ extra colourings with trivial reduction.

  Equation~(\ref{eq:short}) follows from the fact that 
  in every non-trivial colouring in $\adm (\tri,3)$ there must be at least one
  edge coloured $1/2$ and thus $\ker_{\theta'}$ can be at most the number of 
  edges minus one. 

  It follows that for $\beta_1 (\tri, \Z_2)$ or $v$ 
  sufficiently large this bound cannot be tight. For $1$-vertex $\Z_2$-homology
  spheres this bound is sharp as explained in Corollary~\ref{cor:nocol} below.
  Looking at all $1$-vertex triangulations with $\beta_1 (\tri, \Z_2)=1$ up to 
  six tetrahedra, the cases of equality in Inequality~(\ref{eq:short}) are 
  summarised in Table~\ref{tab:small}. See Table~\ref{tab:tv41} for a large 
  number of cases of equality for Inequality~(\ref{eq:long}). 
\qed \end{proof}

\begin{table}[t]
  \begin{center}
    \begin{tabular}{|c|r|r|r|r|r|r|r|}
    \hline
    $(n,\beta_1)$&$\#$ trig.&$\#$ sharp~(\ref{eq:long})&$(r-1)^{n+v}$&$\#$ tree&
    Eqn.~(\ref{eq:short})&Eqn.~(\ref{eq:long})&$\overline{| \adm (\tri, 4) |}$\\
    \hline
    \hline
    $(1,1)$&$1$&$ 1$&$9$&$12.0$&$4$&$4$&$4$ \\
    \hline
    \hline 

    $(2,1)$&$ 5 $&$ 5$&$ 37.8 $&$ 33.0$&$ 10.4 $&$ 6.0 $&$ 6.0 $ \\
    \hline
    $(2,2)$&$ 1 $&$ 1$&$ 27.0 $&$ 39.0$&$ 16.0 $&$ 10.0 $&$ 10.0 $ \\
    \hline
    \hline 

    $(3,1)$&$ 27 $&$ 14$&$ 99.0 $&$ 46.4$&$ 14.7 $&$ 7.7 $&$ 6.3 $ \\
    \hline
    $(3,2)$&$ 3 $&$ 1$&$ 81.0 $&$ 69.0$&$ 28.0 $&$ 14.7 $&$ 11.3 $ \\
    \hline
    \hline 

    $(4,1)$&$ 205 $&$ 67$&$ 378.1$&$ 75.2$&$ 42.9$&$ 13.1$&$ 8.7$ \\
    \hline
    $(4,2)$&$ 19 $&$ 4$&$ 268.6$&$ 110.1$&$ 61.5$&$ 21.8$&$ 15.3$\\
    \hline
    \hline 

    $(5,1)$&$ 1858 $&$ 261$&$ 1131.6$&$ 93.1$&$ 85.5$&$ 20.4$&$ 9.3$ \\
    \hline
    $(5,2)$&$ 184 $&$ 10$&$ 887.5$&$ 159.2$&$ 138.7$&$ 35.8$&$ 18.6$ \\
    \hline
    \hline

    $(6,1)$&$ 21459 $&$ 1574$&$ 3644.8$&$ 120.4$&$ 195.3$&$ 34.6$&$ 10.7$ \\
    \hline 
    $(6,2)$&$ 2516 $&$ 47$&$ 2781.6$&$ 214.5$&$ 297.5$&$ 58.0$&$ 22.0$  \\
    \hline
    $(6,3)$&$ 34 $&$ 0$&$ 2187.0 $&$ 413.2$&$ 456.0 $&$ 94.5$&$ 41.4$  \\
    \hline
    \end{tabular}
    \medskip
    \caption{Analysis of the trivial bound ``$(r-1)^{n+v}$'' (fourth column), average number of
      nodes ``$\#$ tree'' of the search tree visited by the backtracking algorithm (fifth
      column), bound ``Eqn.~(\ref{eq:short})'' given by the Inequality~(\ref{eq:short}) (sixth column), bound ``Eqn.~(\ref{eq:long})'' given by Inequality~(\ref{eq:long}) (seventh column), and average number ``$\overline{| \adm (\tri, 4) |}$'' of 
      admissible colourings in $\adm (\tri, 4)$ (rightmost column).
      The second column lists the number of triangulations ``$\#$ trig.'' contained
      in the census of triangulated closed manifolds $\tri$ with $n$ tetrahedra
      and first Betti number $\beta_1 (\tri, \Z_2)$, the third column 
      lists the number ``$\#$ sharp~(\ref{eq:long})'' of triangulations satisfying equality in 
      Inequality~(\ref{eq:long}).}
    \label{tab:tv41}
  \end{center}
\end{table}


Given a triangulation $\tri$, the right hand side of 
Equation~(\ref{eq:long}) can be computed efficiently.
It turns out to be sharp for $1985$ out of all $26,312$ closed 
triangulated $3$-manifolds with positive first Betti number and up to $6$ 
tetrahedra. In addition, even the average number of colourings is fairly close 
to this bound. See Table~\ref{tab:tv41} for details comparing the first upper bound
to the actual number of colourings in the census. Furthermore, there are $46$ 
triangulations of $3$-manifolds with $\leq 6$ tetrahedra (and positive first 
Betti number) attaining equality in the often much larger right hand side
of Equation~(\ref{eq:short}). For 
details about these cases of equality see Table~\ref{tab:small}. 

We have seen that the number of colourings in $\adm (\tri,3)$ and 
$\adm (\tri,4)$ largely depend on (i) the number of tetrahedra, (ii) the number 
of vertices, and (iii) the first Betti number of $\tri$. Moreover, 
if $\theta \in \adm(\tri,r)$ then  $\theta \in \adm(\tri,r')$ for $r' \geq r$,
and colourings for a lower value of $r$ possibly give rise to an exponential 
number of colours for a higher value of $r$.

Hence, we can expect the number of admissible colourings in $1$-vertex
$\Z_2$-homology spheres to be smaller than in the generic case. Incidentally,
homology spheres are manifolds for which computing $\tv_{r,q}$ is of
particular interest in view of $3$-sphere recognition. For this reason
we have a closer look at this very important special case below.

\subsection{One-vertex $\Z_2$-homology spheres}

When talking about algorithms to compute $\tv_{r,q} (\tri)$ for some 
$3$-manifold triangulation $\tri$, the case of $\Z_2$-homology spheres with 
only one vertex is a special case of particular importance for several reasons.

\begin{enumerate}
  \item One of the most important tasks of $3$-manifold invariants is to 
    distinguish between some $3$-manifold triangulation $\tri$ and the 
    $3$-sphere. In many cases homology can be used to efficiently make this 
    distinction. Hence, this question is most interesting when homology fails,
    that is, when $\tri$ is a homology sphere. 
  \item All results about $\tv_{r,q}$ which apply to 
    $\Z_2$-homology spheres automatically carry through for the invariant 
    $\tv_{r,1}(\tri,[0])$ for arbitrary manifold triangulations $\tri$.
  \item There are powerful techniques available to turn a
    triangulation of a $\Z_2$-homology sphere with an arbitrary number of 
    vertices into a set of smaller triangulations, all with only one vertex, 
    see Section~\ref{sec:algo} for details. 
\end{enumerate}

In this section we take a closer look at $1$-vertex $\Z_2$-homology spheres
and in particular their number of admissible colourings.
One corollary of Proposition~\ref{prop:betav} is the following 
statement for $\Z_2$-homology spheres.

\begin{corollary}
  \label{cor:nocol}
  Let $\tri$ be a $1$-vertex $\Z_2$-homology sphere. Then 
  $| \adm(\tri,r) | = 1$ for $r \leq 4$.
\end{corollary}

\begin{proof}
  For $r=3$ this is a direct consequence from Proposition~\ref{prop:betav},
  for $r=4$ this is a direct consequence from Theorem~\ref{prop:req4}.
%
%
\qed \end{proof}

In particular, $\Z_2$-homology spheres (including many lens spaces) 
can never be distinguished from the $3$-sphere by $\tv_{r,q}$, $r\leq 4$. 

\begin{proposition}
  \label{prop:genbound}
  Let $\tri$ be a $1$-vertex $n$-tetrahedron $\Z_2$-homology sphere. Then 
  for all $\theta \in \adm(\tri,r)$, all colours of $\theta$ are integers,
  and in particular
  $$ | \adm(\tri,r) | \leq \left \lfloor \frac{r}{2} \right \rfloor^{n+1}.$$
\end{proposition}

\begin{proof}
  It follows from Propositions~\ref{prop:reduction} and \ref{prop:betav} that
  no admissible colouring of $\tri$ can contain half-integers. Furthermore,
  the edge colours on every triangle must sum to at most $r-2$ and satisfy the triangle
  inequality. It follows that all colours must be integers between $0$ and
  $\lfloor \frac{r-2}{2} \rfloor$. The statement now follows from the fact that
  $\tri$ has $n+1$ edges. 
\qed \end{proof}

\begin{corollary}
  \label{cor:even}
  Let $\tri$ be a $1$-vertex $n$-tetrahedron closed $3$-manifold triangulation. 
  Then all admissible colourings to compute $\tv_{r,1}(\tri,[0])$ contain 
  integer weights only.
\end{corollary}

\begin{proof}
  Let $\tri$ be a $1$-vertex triangulation. To compute $\tv_{r,1}(\tri,[0])$ 
  we only consider colourings $\theta$ with reductions $\theta'$ corresponding 
  to $1$-coboundaries ($1$-cocycles homologous to $0$). Because $\tri$ has only 
  one vertex, $\theta'$ must be the zero colouring and in particular no 
  half-integers can occur in $\theta$.
\qed \end{proof}

A similar statement for the case of special spines can be found in 
\cite[Remark 8.1.2.2]{matveev03-algms}.

The bound from Proposition~\ref{prop:genbound} cannot be sharp since 
not all triangle colourings $(a,b,c) \in \{ 0, 1, \ldots , \lfloor \frac{r-2}{2} 
\rfloor \}^3$ are admissible. However, for $5 \leq r \leq 7$ we have the 
following situation.

\begin{theorem}
  \label{prop:bounds}
  Let $\tri$ be a $1$-vertex $n$-tetrahedron $\Z_2$-homology $3$-sphere 
  triangulation, then 
  $$ | \adm (\tri,5) | \leq  2^{n}+1; \quad \quad 
     | \adm (\tri,6) |\leq  3^{n}+1; \quad \quad | \adm (\tri,7) | \leq  3^{n}+1.$$
  Moreover, all these upper bounds are sharp.
\end{theorem}

\begin{proof}
  For $r=5$ the admissible triangle colourings are  $(0,0,0)$, $(1/2,1/2,0)$,
  $(1,1,0)$, $(1,1/2,1/2)$, $(1,1,1)$, $(3/2,3/2,0)$, $(3/2,1,1/2)$, up to permutations.
  By Proposition~\ref{prop:reduction}, no colouring in $\adm (\tri,5)$ can
  contain an edge colour $1/2$ or $3/2$. To see this note that otherwise the
  reduction of such a colouring would be a non-trivial colouring in  
  $\adm (\tri,3)$, which does not exist (cf. Proposition~\ref{prop:betav} and
  Corollary~\ref{cor:nocol} with $v=1$ and $\beta_1(\tri,\Z_2) = 0$). Hence, all edge colours must be $0$ or $1$, 
  leaving triangle colourings $(0,0,0)$, $(1,1,0)$, and $(1,1,1)$.

  \medskip
  By an Euler characteristic argument, a $1$-vertex $n$-tetrahedron $3$-manifold has $n+1$ edges. Hence the
  number of colourings of $\tv_{5,q}$ is trivially bounded above by $2^{n+1}$.
  Furthermore, let $\theta \in \adm (\tri,5)$, then either $\theta$ is constant
  $0$ on the edges, constant $1$ on the edges, or $\theta$ contains a triangle 
  coloured $(1,1,0)$.
  In the last case, the complementary colouring $\theta'$, obtained by
  flipping the colour on all the edges, contains a triangle coloured $(0,0,1)$
  and thus $\theta' \not \in \adm (\tri,5)$. It follows that 
  $ | \adm (\tri,5) | \leq  2^{n}+1$.

  \bigskip
  For $r=6$ the admissible triangle colourings are the ones from the case $r=5$
  above plus $(3/2,3/2,1)$, $(2,1,1)$, $(2,2,0)$, $(2,3/2,1/2)$. Again, due to
  Proposition~\ref{prop:reduction}, no half-integers can occur in any colouring.
  Thus, the only admissible triangle colourings are $(0,0,0)$, $(1,1,0)$, 
  $(1,1,1)$, $(2,1,1)$, and $(2,2,0)$.

  We trivially have $| \adm (\tri,6) | \leq 3^{n+1}$. Let 
  $\theta \in \adm (\tri,6)$. We want to show, that at most a third of all
  non-constant assignment of colours $0$, $1$, $2$ to the edges of 
  $\tri$ can be admissible. For this, let $\theta \in \adm (\tri,6)$ and
  let $\theta'$ be
  defined by adding $1$ (mod $3$) to every edge colour. For $\theta'$ to be
  admissible, all triangles of $\theta$ must be of type $(0,0,0)$ and $(2,1,1)$.
  If at least one triangle has colouring $(0,0,0)$, $\theta$ must be the 
  trivial colouring. Hence, all triangles are of type $(2,1,1)$ in $\theta$.
  Replacing $2$ by $0$ and $1$ by $1/2$ in $\theta$ yields a non-trivial
  admissible colouring in $\adm(\tri,3)$, a contradiction by 
  Corollary~\ref{cor:nocol}. Hence, for every
  non-trivial admissible colouring $\theta$, the colouring $\theta'$ cannot
  be admissible. 

  Analogously, let $\theta''$ be defined by adding $2$ (mod $3$) to every edge 
  colour of $\theta$. For $\theta''$ to be admissible, all triangles of $\theta$
  must be of the type $(1,1,1)$, or $(2,2,0)$. A single triangle of type 
  $(1,1,1)$ in $\theta$ forces $\theta$ to be constant. Hence, all triangles
  must be of type $(2,2,0)$. Dividing $\theta$ by four defines a non-trivial
  colouring in $\adm(\tri,3)$, a contradiction. 

  Combining these observations, at most every third non-trivial assignment of 
  colours $0$, $1$, $2$ to the edges of $\theta$ can be admissible. Adding the
  two admissible constant colourings yields $| \adm (\tri,6) | \leq  3^{n}+1$.

  \bigskip
  The proof for $r=7$ follows from a slight adjustment of the proof for $r=6$.
  Admissible triangle colourings for colourings in $\adm(\tri,7)$
  are the ones from $r=6$ plus $(2,2,1)$. Again, we want to show that
  at most every third non-trivial assignment of colours $0$, $1$, $2$ to the
  edges of $\tri$ can be admissible. For this let $\theta \in \adm (\tri,7)$ and
  let $\theta'$ and $\theta''$ be defined as above. For $\theta'$ to be 
  admissible $\theta$ must consist of triangle colourings of type $(0,0,0)$,
  $(1,1,0)$ and $(2,1,1)$. Whenever $\theta$ is non-constant replacing
  $2$ by $0$, and $1$ by $1/2$ yields a non-trivial colouring in 
  $\adm (\tri,3)$ which is not possible. The argument for $\theta''$ is the
  same as in the case $r=6$. It follows that $| \adm (\tri,7) | \leq  3^{n}+1$.

  \bigskip
  All of the above bounds are attained by a number of small $3$-sphere 
  triangulations. See Table~\ref{tab:z2homspheres} for more details about
  $1$-vertex $\Z_2$-homology spheres with up to six tetrahedra and
  their average number of admissible colourings $| \adm (\tri,r) |$, 
  $5 \leq r \leq 7$.
\qed \end{proof}

  \begin{table}[t]
    \begin{center}
      {
      \begin{tabular}{|@{\hspace{0.05cm}}c@{\hspace{0.05cm}}|
        @{\hspace{0.05cm}}r@{\hspace{0.05cm}}|@{\hspace{0.05cm}}r@{\hspace{0.05cm}}||
        @{\hspace{0.05cm}}r@{\hspace{0.05cm}}|@{\hspace{0.05cm}}r@{\hspace{0.05cm}}|
        @{\hspace{0.05cm}}r@{\hspace{0.05cm}}||@{\hspace{0.05cm}}r@{\hspace{0.05cm}}|
        @{\hspace{0.05cm}}r@{\hspace{0.05cm}}|@{\hspace{0.05cm}}r@{\hspace{0.05cm}}||
        @{\hspace{0.05cm}}r@{\hspace{0.05cm}}|@{\hspace{0.05cm}}r@{\hspace{0.05cm}}|
        @{\hspace{0.05cm}}r@{\hspace{0.05cm}}|}
        \hline
        $n$&$\#$trig.&$\#$sharp&$(5\!-\!1)^{n+v}$&$2^n\!\!+\!\!1$&
        {\scriptsize$\overline{| \adm (\tri,\!5) |}$}&$(6\!-\!1)^{n+v}$&$3^n\!\!+\!\!1$&
        {\scriptsize$\overline{| \adm (\tri,\!6) |}$}&$(7\!-\!1)^{n+v}$&$3^n\!\!+\!\!1$&{\scriptsize$\overline{| \adm (\tri,\!7) |}$}\\
        \hline
        \hline
        $1$&$ 2 $&$ 1$&$ 16$&$3$&$ 2.50 $&$ 25$&$4 $&$ 3.00 $&$ 36$&$ 4$&$ 4.00 $ \\
        \hline
        $2$&$ 7 $&$ 3$&$ 64$&$5$&$ 4.00 $&$ 125$&$10 $&$ 6.86 $&$ 216$&$ 10$&$ 8.86 $ \\
        \hline
        $3$&$ 36 $&$ 5$&$ 256$&$9$&$ 5.61 $&$ 625$&$28 $&$ 12.22 $&$ 1,296$&$ 28$&$ 17.28 $ \\
        \hline
        $4$&$ 255 $&$ 14$&$ 1,024$&$17$&$ 8.31 $&$ 3,125$&$82 $&$ 23.46 $&$ 7,776$&$ 82$&$ 35.30 $ \\
        \hline
        $5$&$ 2305 $&$ 30$&$ 4,096$&$33$&$ 12.02 $&$ 15,625$&$244 $&$ 43.00 $&$ 46,656$&$ 244$&$ 70.44 $ \\
        \hline
        $6$&$ 24597 $&$ 89$&$ 16,384$&$65$&$ 17.71 $&$ 78,125$&$730 $&$ 80.15 $&$ 279,936$&$ 730$&$ 142.23 $ \\
        \hline
      \end{tabular} }
      \medskip
      \caption{
      Number ``$\#$sharp'' 
      of $1$-vertex $\Z_2$-homology spheres with $n$ tetrahedra, $1 \leq n \leq 6$, 
      satisfying equality in all bounds from Theorem~\ref{prop:bounds}
      (third column), and the average number ``$\overline{\adm (\tri, r)}$'' of 
      admissible colourings in $\adm (\tri, r)$, $5 \leq r \leq 7$ (columns 
      $6$, $9$, and $12$), compared to the naive upper bound ``$(r-1)^{v+n}$'' (columns $4$, $7$, and $10$) and the new upper bounds given by Theorem~\ref{prop:bounds} (columns $5$, $8$, and $11$).}
      \label{tab:z2homspheres}
    \end{center}
  \end{table}


There are $27,202$ $\Z_2$-homology spheres with $1$-vertex and up to $6$ tetrahedra.
Exactly $142$ of them attain equality in all three bounds. For more details 
about these cases of equality and the average number of colourings for 
$5 \leq r \leq 7$ in the census, see Table~\ref{tab:z2homspheres}.

Note that the sharp bounds from Theorem~\ref{prop:bounds} suggest that 
the over count of the general bound from Proposition~\ref{prop:genbound} is only 
linear in $r$.

%
%
%
%
%


\section{Faster ways to compute $\tv_{r,q}$}
\label{sec:algo}

In this section we describe an algorithm to compute
$\tv_{4,q}$---a problem known to be $\#P$-hard---exploiting the combinatorial
structure of the input triangulation. Moreover, we describe a significant 
exponential speed-up for computing $\tv_{r,1} (\tri)$ in the case where $r$ odd. 
However, before we can describe the new algorithms, we first have to briefly 
state some classical results about Turaev-Viro invariants.

\subsection{Classical results about Turaev-Viro invariants}
\label{app:algo}

Note that the Turaev-Viro invariants $\tv_{r,q}$ are closely related to 
the more general invariant of Witten and Reshetikhin-Turaev 
$\tau_{r,q}$ ($\in \mathbb{C}$), due to the following result.

\begin{theorem}[Turaev \cite{turaev10-book}, Roberts \cite{roberts95-skein}]
  \label{thm:wrt}
  For the invariants of Witten and Reshetikhin-Turaev $\tau_{r,q}$, 
  and the Turaev-Viro invariants 
  $$ \tv_{r,q} = \mid \tau_{r,q} \mid^2 $$
  holds.
\end{theorem}

Theorem~\ref{thm:wrt} enables us to translate a number of key results about
the Witten and Reshetikhin-Turaev invariants in terms of Turaev-Viro invariants.
Namely, the following statement holds.

\begin{theorem}[Based on Kirby and Melvin~\cite{kirby91-witten}]
  \label{thm:prod}
  Let $M$ and $N$ be closed compact $3$-manifolds, and let $r \geq 3$,
  $1 \leq q \leq r-1$. Then there exist $\gamma_r \in \mathbb{C}$, such that
  for $\tv_{r,1}' = \gamma_r \tv_{r,1}$ we have
  $$ \tv_{r,1}' (M \# N) = \tv_{r,1}' (M) \cdot \tv_{r,1}' (N). $$
\end{theorem}
Additionally, when a manifold $M$ is represented by a triangulation with $n$ tetrahedra, the normalising factor $\gamma_r$ can be computed in polynomial time in $n$.

Using Turaev-Viro invariants at the trivial cohomology class we have the
following identity for odd degree $r$.

\begin{theorem}[Based on Kirby and Melvin~\cite{kirby91-witten}]
  \label{thm:rodd}
  Let $M$ be a closed compact $3$-manifold, and let $r\geq 3$ be an odd integer.
  Then $$ \tv_{r,1} (M) = \tv_{3,1} (M) \cdot \tv_{r,1}(M,[0]). $$
\end{theorem}

\subsection{A structure-sensitive algorithm to compute $\adm(\tri,4)$}

The algorithm we present in this section is a direct consequence of
the proof of Theorem~\ref{prop:req4}.

\medskip
\noindent
{\bf Input}: A $v$-vertex $n$-tetrahedra triangulation of a closed $3$-manifold 
$\tri$ with set of edges $E$

\smallskip
\noindent
{\bf 1.}: Compute $\adm(\tri,3)$. Furthermore, for all 
$\theta \in \adm(\tri,3)$, enumerate the set of edges $\ker_\theta \subset E$ of
$\tri$ coloured zero in $\theta$.

\smallskip
\noindent
{\bf 2.}: For each non-trivial $\theta \in \adm(\tri,3)$, for each subset 
$A \setminus \ker_\theta$: Let $\theta'$ be the edge colouring that colours
(i) all edges in $A$ by $1$, (ii) all edges in ($E \setminus \ker_\theta$) by 
$1/2$, and (iii) all edges in ($\ker_\theta \setminus A$) by $0$.
For each non-trivial $\theta$, set up a backtracking procedure to check all such
$\theta'$ for admissibility. Add the admissible colourings $\theta'$ to 
$\adm (\tri,4)$.

\smallskip
\noindent
{\bf 3.}: For all colourings $\theta \in \adm(\tri,3)$, double all colours of
$\theta$ and add the result to $\adm (\tri,4)$.

\medskip
\noindent
{\bf Correctness of the algorithm and experiments:} 
%
Due to Theorem~\ref{prop:req4} we know that the above procedure enumerates all 
colourings in $\adm (\tri,4)$. Computing $\tv_{4,q}(\tri)$ thus runs in 
$$O \left ( \left ( \Sigma_{ \theta \in \adm(\tri,3) \setminus \{ 0 \}} 
   \,\, 2^{\ker_{\theta}} \right ) + 2^{v+\beta_1(\tri,\Z_2)-1} \right ) $$
arithmetic operations in $\Q[e^{i q \pi / 4}]$. This upper bound is much smaller
than the worst case running time $(r-1)^{n+v}$ of the
naive backtracking procedure. However, the backtracking algorithm typically
performs much better than this pathological upper bound.

Hence, one straightforward question to ask is (i) compared to the worst case
running time of the new algorithm, how many nodes of the full
search tree are actually visited by the naive backtracking algorithm, and (ii) 
how close is the worst case running time of the new algorithm to the actual 
number of admissible colourings
of typical inputs. To give a partial answer to this question we analyse the 
census of closed triangulations up to six tetrahedra. For every $v$-vertex, 
$n$-tetrahedra triangulation $\tri$ we compare the naive bound $3^{n+v}$, the 
size of the search tree traversed by the backtracking algorithm specific
to $\tri$, the 
improved general bound from Equation~(\ref{eq:short}),
the bound specific to $\tri$ from Equation~(\ref{eq:long}),
and the actual number of admissible colourings $| \adm (\tri , 4) |$. As a 
result we find that (i) the actual number of nodes visited by the backtracking 
algorithm is small but still significantly larger than the upper bound given in 
Equation~(\ref{eq:long}), and (ii) the right hand side of Equation~(\ref{eq:long}) 
is surprisingly close to $| \adm(\tri,4) |$. A summary containing the average 
values of the bounds over all triangulations with fixed number of tetrahedra and
$\Z_2$-Betti number can be found in Table~\ref{tab:tv41}.

\subsection{An algorithm to compute $\tv_{r,1}$, $r$ odd}

In this section we describe a significant exponential speed-up 
for computing $\tv_{r,1} (\tri)$ in the case where $r$ is odd and
$\tri$ does not contain any two-sided projective planes~\footnote{This is a technical pre-condition for the crushing procedure to succeed. Triangulations on which the crushing procedure fails are however extremely rare.}. The main ingredients 
for this speed-up are:

\begin{itemize}
  \item The crushing and expanding procedure for closed $3$-manifolds as 
    described by Burton, and Burton and Ozlen, which turns an arbitrary 
    $v$-vertex triangulation into a number of smaller $1$-vertex triangulations
    in polynomial time \cite{burton14-crushing-dcg,burton12-unknot};
  \item A classical result about Turaev-Viro invariants due to Turaev 
    \cite{turaev10-book}, Roberts \cite{roberts95-skein}, and Kirby and 
    Melvin~\cite{kirby91-witten} stating that there exist a scaled version 
    $\tv_{r,1}'= \gamma_r \tv_{r,1}$ which is multiplicative under taking 
    connected sums (see Theorem~\ref{thm:prod});
  \item Another classical result due to the same authors and publications
    stating that, for $r$ odd, we have $\tv_{r,1}(\tri) = \tv_{3,1} (\tri) \cdot 
    \tv_{r,1}(\tri,[0])$, and thus $\tv_{r,1}(\tri,[0])$ is essentially
    sufficient to
    compute $\tv_{r,1} (\tri)$ (see Theorem~\ref{thm:rodd});
  \item Corollary~\ref{cor:even} stating that computing $\tv_{r,1}(\tri,[0])$
    of a $1$-vertex closed $3$-manifold triangulation can be done by only 
    enumerating colourings with all integer colours.
\end{itemize}

\begin{figure}[t]
  \begin{center}
    \includegraphics[width=\textwidth]{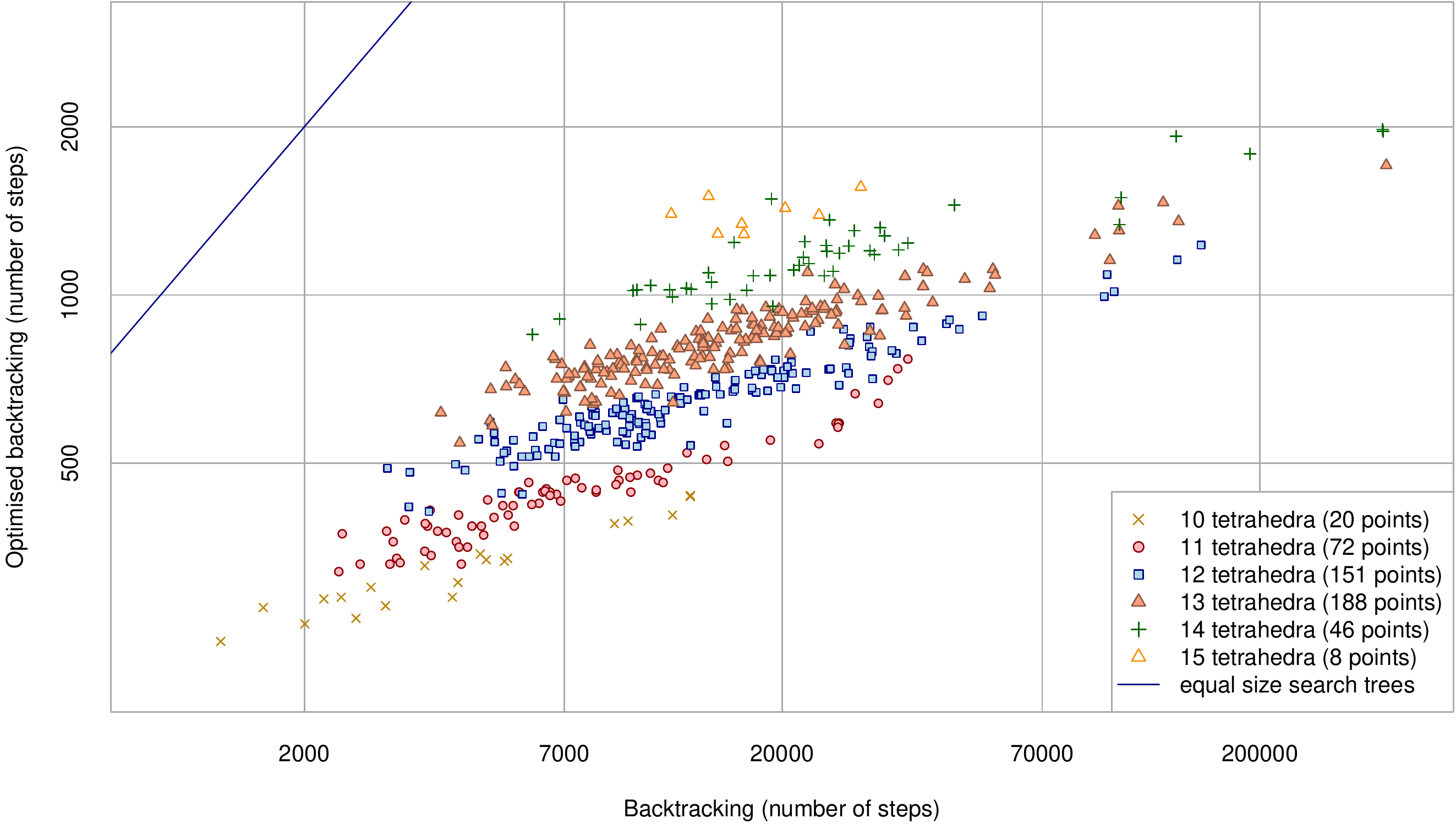}
  \end{center}
  \caption{Number of nodes in the search tree visited by the naive algorithm and the optimised backtracking procedure for the $500$ first $1$-vertex triangulations of the Hodgson-Weeks census.\label{fig:expts1}}
\end{figure}

\medskip
\noindent
{\bf Input}: A $v$-vertex $n$-tetrahedra triangulation of a closed $3$-manifold 
$\tri$

%
%

\smallskip
\noindent
{\bf 1.}: If $\tri$ has more than one vertex, apply the crushing and expanding
procedure to $\tri$ as described in \cite{burton14-crushing-dcg} and 
\cite{burton12-unknot} respectively. As a result
we obtain a number of triangulations $\tri_1, \ldots , \tri_m$, and a number of 
``removed components'' $c_1 , \ldots , c_{\ell}$ with the following properties.

\begin{itemize}
  \item Every triangulation $\tri_i$, $1 \leq i \leq m$, is a $1$-vertex 
    triangulation;
  \item If $n_i$ is the number of tetrahedra in $\tri_i$, then
    $\ell + \sum \limits_{i=0}^{m} n_i \leq n; $
  \item Every ``removed component'' is either a $3$-sphere, the lens space 
    $L(3,1)$, or the real projective space $\mathbb{R}P^3$. Note that 
    homology calculations can distinguish between all of these pieces in polynomial time;
  \item We have for the topological type of the $\tri_i$ that $\tri$ is the \emph{connected sum}\footnote{Building the connected sum $M \# N$ of two manifolds $M$ and $N$ simply consists of removing a small ball from $M$ and $N$ respectively, and glue them together along their newly created boundaries.}:
    \begin{equation}
      \label{eq:connsum}
      \tri \cong \tri_1 \# \ldots \# \tri_m \# c_1 \# \ldots \# c_{\ell}.
    \end{equation}
\end{itemize}

If $\tri$ contains a two-sided projective plane the crushing procedure will 
detect this fact and the computation is cancelled. The total running time of
this step is polynomial.

\smallskip
\noindent
{\bf 2.}: Compute $\tv_{r,1}(\tri_i,[0])$, $1 \leq i \leq m$. 

\smallskip
\noindent
{\bf 3.}: For all $\tri_i$, compute $\tv_{3,1} (\tri_i)$, and use 
Theorem~\ref{thm:rodd} to obtain $\tv_{r,1} (\tri_i)$. The Turaev-Viro
invariants of $S^3$, $\mathbb{R}P^3$, and $L(3,1)$ are well known (see Sokolov 
\cite{Sokolov97LensSpacesTVInv}) and the respective values for the $c_i$ can 
efficiently be pre-computed. If one of the components $c_i$ is a 
real projective space, return $0$, as $\tv_{r,1} (\mathbb{R}P^3) = 0$ for all
$r \geq 3$, odd.

\smallskip
\noindent
{\bf 4.}: Scale all values from the previous step to $\tv_{r,1}'$, multiply them
and re-scale the product. The result equals $\tv_{r,1}(\tri)$, by 
Theorem~\ref{thm:prod} and Equation~(\ref{eq:connsum}).

%
%

\begin{figure}[t]
  \begin{center}
    \includegraphics[width=\textwidth]{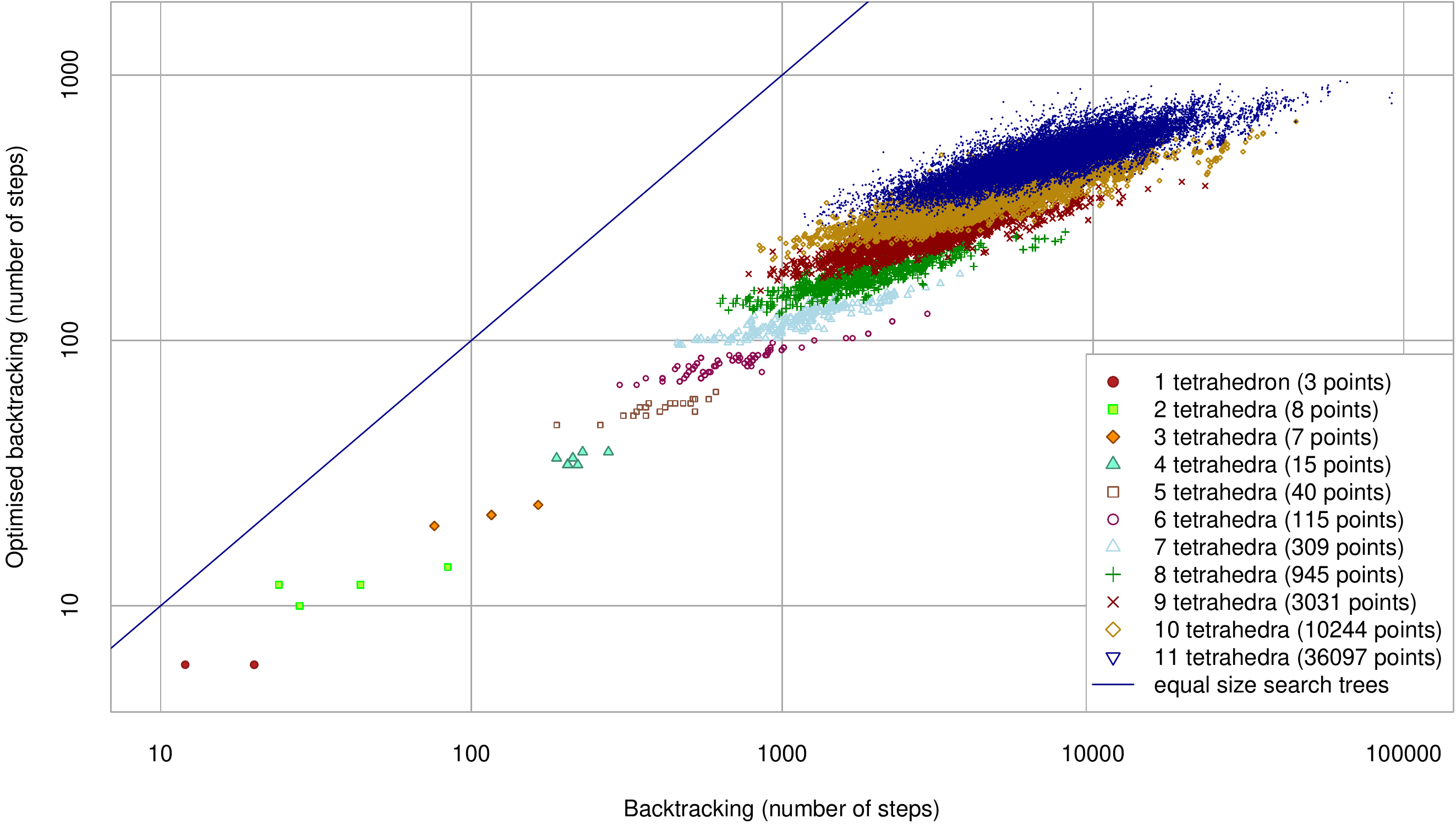}
  \end{center}
  \caption{Number of nodes in the search tree visited by the optimised backtracking procedure over the naive algorithm for the $50,814$ $1$-vertex minimal closed triangulations.
  \label{fig:expts2}}
\end{figure}

\medskip
\noindent
{\bf Running time, efficiency and effectiveness of the proposed algorithm:}
%
The simplifying step, the crushing and expanding procedure, and computing
$\tv_{3,1} (\tri)$ are all polynomial time algorithms 
\cite{burton14-crushing-dcg,burton12-unknot}. Following Corollary~\ref{cor:even}
and Proposition~\ref{prop:genbound} the running time to compute 
$\tv_{r,1}(\tri_i,[0])$ is $O(\left \lfloor r/2 \right \rfloor^{n_i+1} )$ 
(remember, $\tri_1$ is a $1$-vertex triangulation). The overall running time is 
thus $O(\left \lfloor r/2 \right \rfloor^{n+1})$.
The same procedure can be applied to improve the fixed parameter tractable algorithm as presented 
in \cite{Burton15TuraevViro}---which is also based on enumerating colourings---to get the running time $O (n \lfloor r/2 \rfloor^{6(k+1)} k^2 \log r)$, where $k$ is the treewidth of the dual graph of $\tri$.


To compare the performance of the naive backtracking with 
the proposed algorithm, we count the number of nodes in the search
tree visited by both algorithms for computing $\tv_{5,1}$ (i) for the first $500$ triangulations of the
Hodgson-Weeks census, with $10 \leq n \leq 15$, \cite{hodgson94-closedhypcensus}, and (ii) for all 
triangulations with $\leq 11$ tetrahedra in the census of closed minimal 
triangulations \cite{burton11-genus}, see Figure~\ref{fig:expts1} and
\ref{fig:expts2}. These triangulations all have $1$ vertex, and the improvement is solely due to the reduction of the space of colourings studied in this section (in particular, the crushing step is not applied). 
Improvements in the minimal triangulations census range from factors $2$ to 
$117$. Improvements in the Hodgson-Weeks census, which contains larger 
triangulations, range from factors $5.6$ to  $215$. Observe how the range of
improvements rapidly grows larger as the size of the triangulations increase.

As evidence for the effectiveness of the algorithm, we analyse the ability of 
$\tv_{r,1}$, $r \in \{3,5,7,9\}$, to distinguish $3$-manifolds from the 
$3$-sphere $S^3$. Since homology can be computed in polynomial time, we only 
consider $\Z$-homology spheres, i.e., $3$-manifolds with the homology groups of 
$S^3$. There are $36$ distinct $\Z$-homology spheres of {\em complexity} 
at most $11$, meaning, they can be triangulated with $11$ tetrahedra or less. 
Due to Corollary~\ref{cor:nocol} we already know that none of them can be 
distinguished from $S^3$ by $\tv_{3,1}$ (note that $\Z$-homology spheres are 
always $\Z_2$-homology spheres). $\tv_{5,1}$, $\tv_{7,1}$, and $\tv_{9,1}$,
distinguish $22$, $31$, and $30$ of them from $S^3$. Furthermore, a combination 
of $\tv_{5,1}$ and $\tv_{7,1}$ only fails once, and a combination of all three 
invariants never fails to distinguish $\Z$-homology spheres of complexity
$\leq 11$ from $S^3$. See Tables~\ref{tab:Zhomsp} and \ref{tab:appZhomsp} for 
details.

\begin{table}[t]
  \begin{center}
    \begin{tabular}{|r|c|c|c|c|c|c|c|c|}
    \hline
    $n$&$\tv_{3,1}$&$\tv_{4,1}$&$\tv_{5,1}$&$\tv_{6,1}$&$\tv_{7,1}$&$\tv_{8,1}$&$\tv_{9,1}$&$\tv_{5,1}$ and $\tv_{7,1}$ \\
    \hline
    \hline
    $5$& 0/1& 0/1& {\bf 1/1}& 0/1& {\bf 1/1}& 1/1& 1/1& {\bf 1/1} \\
    \hline
    $7$& 0/1& 0/1& {\bf 1/1}& 0/1& {\bf 1/1}& 0/1& 1/1& {\bf 1/1} \\
    \hline
    $8$& 0/3& 0/3& {\bf 1/3}& 0/3& {\bf 3/3}& 3/3& 3/3& {\bf 3/3} \\
    \hline
    $9$& 0/4& 0/4& {\bf 3/4}& 0/4& {\bf 3/4}& 1/4& 3/4& {\bf 4/4} \\
    \hline
    $10$& 0/8& 0/8& {\bf 5/8}& 0/8& {\bf 7/8}& 3/8& 6/8& {\bf 8/8} \\
    \hline
    $11$& 0/19& 0/19& {\bf 11/19}& 0/19& {\bf 16/19}& 13/19& 16/19& {\bf 18/19} \\
    \hline
    \end{tabular}
    \medskip
    \caption{Summary of the ability of $\tv_{r,1}$, $3 \leq r \leq 9$, to 
      distinguish $\Z$-homology spheres up to complexity $11$ from the 
      $3$-sphere. X/Y denotes the success rate, i.e., there are Y 
      manifolds, X of which can be distinguished from the $3$-sphere by the 
      respective invariant.}
    \label{tab:Zhomsp}
  \end{center}
\end{table}

\begin{table}
  \begin{center}
    \begin{tabular}{|r|l|@{}c@{}|@{}c@{}|@{}c@{}|@{}c@{}|@{}c@{}|@{}c@{}|@{}c@{}|}
      \hline
%
%
%
%
%

      $n$&top. type&$\tv_{3,1}$&$\tv_{4,1}$&$\tv_{5,1}$&$\tv_{6,1}$&$\tv_{7,1}$&$\tv_{8,1}$&$\tv_{9,1}$ \\ 
      \hline
      \hline 
      $5$&$\Sigma(2,3,5)$&$0$&$0$&$1$&$0$&$1$&$1$&$1$\\ 
      \hline
      \hline

      $7$&$\Sigma(2,3,7)$&$0$&$0$&$1$&$0$&$1$&$0$&$1$\\ 
      \hline
      \hline

      $8$&$\Sigma(2,3,11)$&$0$&$0$&$0$&$0$&$1$&$1$&$1$\\ 
      \hline
      &$\Sigma(3,4,5)$&$0$&$0$&$0$&$0$&$1$&$1$&$1$\\ 
      \hline
      &$\operatorname{SFS} [D: (2,1) (3,1)] \cup_{0,1 | 1,0} \operatorname{SFS} [D: (2,1) (3,2)]$&$0$&$0$&$1$&$0$&$1$&$1$&$1$\\ 
      \hline
      \hline

      $9$&$\Sigma(2,3,13)$&$0$&$0$&$1$&$0$&$0$&$1$&$1$\\ 
      \hline
      &$\Sigma(2,3,17)$&$0$&$0$&$1$&$0$&$1$&$0$&$0$\\ 
      \hline
      &$\Sigma(2,5,7)$&$0$&$0$&$1$&$0$&$1$&$0$&$1$\\ 
      \hline
      &$\Sigma(3,4,7)$&$0$&$0$&$0$&$0$&$1$&$0$&$1$\\ 
      \hline
      \hline

      $10$&$\Sigma(2,3,19)$&$0$&$0$&$0$&$0$&$1$&$1$&$0$\\ 
      \hline
      &$\Sigma(2,3,23)$&$0$&$0$&$1$&$0$&$1$&$0$&$1$\\ 
      \hline
      &$\Sigma(2,7,9)$&$0$&$0$&$0$&$0$&$1$&$0$&$1$\\ 
      \hline
      &$\Sigma(3,5,11)$&$0$&$0$&$0$&$0$&$1$&$0$&$1$\\ 
      \hline
      &$\Sigma(3,7,8)$&$0$&$0$&$1$&$0$&$0$&$0$&$0$\\ 
      \hline
      &$\operatorname{SFS} [D: (2,1) (3,1)] \cup_{0,1 | 1,0} \operatorname{SFS} [D: (2,1) (7,5)]$&$0$&$0$&$1$&$0$&$1$&$0$&$1$\\ 
      \hline
      &$\operatorname{SFS} [D: (2,1) (5,2)] \cup_{0,1 | 1,0} \operatorname{SFS} [D: (2,1) (5,3)]$&$0$&$0$&$1$&$0$&$1$&$1$&$1$\\ 
      \hline
      &$\operatorname{Hyp}1.39850888$&$0$&$0$&$1$&$0$&$1$&$1$&$1$\\ 
      \hline
      \hline

      $11$&$\Sigma(2,3,25)$&$0$&$0$&$1$&$0$&$1$&$0$&$1$\\ 
      \hline
      &$\Sigma(2,3,29)$&$0$&$0$&$0$&$0$&$0$&$1$&$1$\\ 
      \hline
      &$\Sigma(2,5,13)$&$0$&$0$&$1$&$0$&$0$&$1$&$1$\\ 
      \hline
      &$\Sigma(2,5,17)$&$0$&$0$&$1$&$0$&$1$&$0$&$0$\\ 
      \hline
      &$\Sigma(2,9,11)$&$0$&$0$&$0$&$0$&$1$&$0$&$1$\\ 
      \hline
      &$\Sigma(3,4,17)$&$0$&$0$&$0$&$0$&$1$&$0$&$0$\\ 
      \hline
      &$\Sigma(3,4,19)$&$0$&$0$&$0$&$0$&$1$&$1$&$0$\\ 
      \hline
      &$\Sigma(3,7,13)$&$0$&$0$&$1$&$0$&$0$&$0$&$1$\\ 
      \hline
      &$\operatorname{SFS} [D: (2,1) (3,1)] \cup_{-1,2 | 0,1} \operatorname{SFS} [D: (3,2) (5,3)]$&$0$&$0$&$0$&$0$&$1$&$1$&$1$\\ 
      \hline
      &$\operatorname{SFS} [D: (2,1) (3,2)] \cup_{-2,3 | -1,2} \operatorname{SFS} [D: (2,1) (5,2)]$&$0$&$0$&$0$&$0$&$1$&$1$&$1$\\ 
      \hline
      &$\operatorname{SFS} [D: (2,1) (3,2)] \cup_{0,1 | 1,0} \operatorname{SFS} [D: (2,1) (11,4)]$&$0$&$0$&$0$&$0$&$1$&$1$&$1$\\ 
      \hline
      &$\operatorname{SFS} [D: (2,1) (3,2)] \cup_{0,1 | 1,0} \operatorname{SFS} [D: (3,2) (5,1)]$&$0$&$0$&$1$&$0$&$1$&$1$&$1$\\ 
      \hline
      &$\operatorname{SFS} [D: (2,1) (3,2)] \cup_{0,1 | 1,0} \operatorname{SFS} [D: (4,1) (5,3)]$&$0$&$0$&$1$&$0$&$1$&$1$&$1$\\ 
      \hline
      &$\operatorname{SFS} [D: (2,1) (5,3)] \cup_{0,1 | 1,0} \operatorname{SFS} [D: (3,2) (4,1)]$&$0$&$0$&$0$&$0$&$1$&$1$&$1$\\ 
      \hline
      &$\operatorname{SFS} [D: (2,1) (3,1)] \cup_{-1,1 | 1,0} \operatorname{SFS} [A: (2,1)] \cup_{0,1 | 1,1} $&$0$&$0$&$1$&$0$&$1$&$1$&$1$\\
      &$\operatorname{SFS} [D: (2,1) (3,2)]$&&&&&&&\\
      \hline
      &$\operatorname{Hyp}1.73198278$&$0$&$0$&$1$&$0$&$1$&$0$&$1$\\ 
      \hline
      &$\operatorname{Hyp}1.91221025$&$0$&$0$&$1$&$0$&$1$&$1$&$1$\\ 
      \hline
      &$\operatorname{Hyp}2.22671790$&$0$&$0$&$1$&$0$&$1$&$1$&$1$\\ 
      \hline
      &$\operatorname{Hyp}2.25976713$&$0$&$0$&$1$&$0$&$1$&$1$&$1$ \\
      \hline
    \end{tabular}
    \medskip
    \caption{$\Z$-homology $3$-spheres up to complexity $11$ and the ability
      of $\tv_{r,1}$, $r\leq 9$ to distinguish them from the $3$-sphere ($0 =$ no
      distinction, $1 =$ distinction).
      \label{tab:appZhomsp}}
  \end{center}
\end{table}

\bibliographystyle{plain}
\bibliography{pure}

\begin{thebibliography}{10}

\bibitem{Bachmann12Helical}
David Bachman.
\newblock Normalizing topologically minimal surfaces ii: Disks.
\newblock {\em arXiv:1210.4574}, 2012.

\bibitem{birman1975braids}
Joan~S. Birman.
\newblock {\em Braids, links, and mapping class groups}.
\newblock Annals of mathematics studies. Princeton University Press, 1975.

\bibitem{burton11-genus}
Benjamin~A. Burton.
\newblock Detecting genus in vertex links for the fast enumeration of
  3-manifold triangulations.
\newblock In {\em Proceedings of {ISSAC}}, pages 59--66. ACM, 2011.

\bibitem{burton14-crushing-dcg}
Benjamin~A. Burton.
\newblock A new approach to crushing 3-manifold triangulations.
\newblock {\em Discrete Comput. Geom.}, 52(1):116--139, 2014.

\bibitem{regina}
Benjamin~A. Burton, Ryan Budney, Will Pettersson, et~al.
\newblock Regina: Software for 3-manifold topology and normal surface theory.
\newblock \texttt{http://\allowbreak regina.\allowbreak sourceforge.\allowbreak
  net/}, 1999--2014.

\bibitem{Burton15TuraevViro}
Benjamin~A. Burton, Cl\'ement Maria, and Jonathan Spreer.
\newblock {Algorithms and complexity for Turaev-Viro invariants}.
\newblock In {\em Proceedings of ICALP 2015}, pages 281--293. Springer, 2015.

\bibitem{burton12-unknot}
Benjamin~A. Burton and Melih Ozlen.
\newblock A fast branching algorithm for unknot recognition with experimental
  polynomial-time behaviour.
\newblock {\em arXiv:1211.1079}, 2012.

\bibitem{haken61-knot}
Wolfgang Haken.
\newblock Theorie der Normalfl{\"a}chen.
\newblock {\em Acta Math.}, 105:245--375, 1961.

\bibitem{hass99-knotnp}
Joel Hass, Jeffrey~C. Lagarias, and Nicholas Pippenger.
\newblock The computational complexity of knot and link problems.
\newblock {\em J. Assoc. Comput. Mach.}, 46(2):185--211, 1999.

\bibitem{hatcher02-algebraic}
Allen Hatcher.
\newblock {\em Algebraic Topology}.
\newblock Cambridge University Press, Cambridge, 2002.
\newblock \texttt{http://\allowbreak www.\allowbreak math.\allowbreak
  cornell.\allowbreak edu/\allowbreak \~{}hatcher/\allowbreak AT/\allowbreak
  ATpage.\allowbreak html}.

\bibitem{hodgson94-closedhypcensus}
Craig~D. Hodgson and Jeffrey~R. Weeks.
\newblock Symmetries, isometries and length spectra of closed hyperbolic
  three-manifolds.
\newblock {\em Experiment. Math.}, 3(4):261--274, 1994.

\bibitem{kirby91-witten}
Robion Kirby and Paul Melvin.
\newblock The {$3$}-manifold invariants of {W}itten and {R}eshetikhin-{T}uraev
  for {${\rm sl}(2,{\bf C})$}.
\newblock {\em Invent. Math.}, 105(3):473--545, 1991.

\bibitem{kirby04-nphard}
Robion Kirby and Paul Melvin.
\newblock Local surgery formulas for quantum invariants and the {A}rf
  invariant.
\newblock {\em Geom. Topol. Monogr.}, pages (7):213--233, 2004.

\bibitem{kleiner08-perelman}
Bruce Kleiner and John Lott.
\newblock Notes on {P}erelman's papers.
\newblock {\em Geom. Topol.}, 12(5):2587--2855, 2008.

\bibitem{matveev03-algms}
Sergei Matveev.
\newblock {\em Algorithmic Topology and Classification of 3-Manifolds}.
\newblock Number~9 in Algorithms and Computation in Mathematics. Springer,
  Berlin, 2003.

\bibitem{recogniser}
Sergei Matveev et~al.
\newblock Manifold recognizer.
\newblock \texttt{http://\allowbreak www.\allowbreak matlas.\allowbreak
  math.\allowbreak csu.\allowbreak ru/\allowbreak ?page=\allowbreak
  recognizer}, accessed August 2012.

\bibitem{roberts95-skein}
Justin Roberts.
\newblock Skein theory and {T}uraev-{V}iro invariants.
\newblock {\em Topology}, 34(4):771--787, 1995.

\bibitem{rubinstein97-3sphere}
J.~Hyam Rubinstein.
\newblock Polyhedral minimal surfaces, {H}eegaard splittings and decision
  problems for {$3$}-di\-men\-sion\-al manifolds.
\newblock In {\em Geometric Topology}, volume~2 of {\em AMS/IP Stud. Adv.
  Math.}, pages 1--20. Amer. Math. Soc., 1997.

\bibitem{Sokolov97LensSpacesTVInv}
M.~V. Sokolov.
\newblock Which lens spaces are distinguished by {T}uraev-{V}iro invariants?
\newblock {\em {Mathematical Notes}}, 61(3):468--470, 1997.

\bibitem{Stillwell93Undecidability}
John Stillwell.
\newblock {\em Classical topology and combinatorial group theory}, volume~72 of
  {\em Graduate Texts in Mathematics}.
\newblock Springer-Verlag, New York, second edition, 1993.

\bibitem{turaev10-book}
Vladimir~G. Turaev.
\newblock {\em Quantum Invariants of Knots and 3-Manifolds}, volume~18 of {\em
  de Gruyter Studies in Mathematics}.
\newblock Walter de Gruyter \& Co., Berlin, revised edition, 2010.

\bibitem{turaev92-invariants}
Vladimir~G. Turaev and Oleg~Y. Viro.
\newblock State sum invariants of {$3$}-manifolds and quantum {$6j$}-symbols.
\newblock {\em Topology}, 31(4):865--902, 1992.

\end{thebibliography}


\end{document}